\documentclass[11pt]{article}
\usepackage{fullpage}
\usepackage[T1]{fontenc}
\usepackage[utf8]{inputenc}
\usepackage{xcolor}
\usepackage[colorlinks=true]{hyperref}
\definecolor{purple}{HTML}{7068A4}
\definecolor{blue}{HTML}{3F70B2}
\definecolor{orange}{HTML}{A47458}
\definecolor{yellow}{HTML}{9B7A3C}
\hypersetup{
    linkcolor=blue
    ,citecolor=orange
    ,filecolor=purple
    ,urlcolor=orange
    ,menucolor=purple
    ,runcolor=purple
}
\usepackage{amsmath, amssymb, amsthm, amsfonts, graphicx, color, subcaption, enumerate, bm, cleveref, enumitem, array, mathtools}
\usepackage[colorinlistoftodos,prependcaption,textsize=tiny]{todonotes}

\usepackage{multicol, multirow}
\usepackage{thmtools, thm-restate}
\usepackage{ragged2e}
\usepackage{lineno}
\usepackage{framed}
\usepackage[framemethod=tikz]{mdframed}

\usepackage[ruled,linesnumbered]{algorithm2e}
\SetKwComment{Comment}{/* }{ */}

\usepackage{subfiles}
\graphicspath{{graphics/}}

\crefname{claim}{Claim}{Claims}
\crefname{property}{Property}{Properties}
\crefname{algocf}{Algorithm}{Algorithms}
\Crefname{algocf}{Algorithm}{Algorithms}

\usepackage[style=trad-alpha,natbib=true,maxcitenames=4]{biblatex}
\makeatletter
\g@addto@macro\bfseries{\boldmath}
\makeatother

\newtheorem{lemma}{Lemma}[section]
\newtheorem{claim}[lemma]{Claim}

\newtheorem{corollary}[lemma]{Corollary}
\newtheorem{observation}[lemma]{Observation}
\newtheorem{proposition}[lemma]{Proposition}

\theoremstyle{definition}
\newtheorem{definition}[lemma]{Definition}

\theoremstyle{remark}

\newtheorem*{remark*}{Remark}

\newcommand{\onlineLOCAL}{\mathsf{Online}$-$\mathsf{LOCAL}}
\newcommand{\dynamicLOCAL}{\mathsf{Dynamic}$-$\mathsf{LOCAL}}
\newcommand{\dynamicLOCALpm}{\mathsf{Dynamic}$-$\mathsf{LOCAL}^\pm}
\newcommand{\LOCAL}{\mathsf{LOCAL}}
\newcommand{\SLOCAL}{\mathsf{SLOCAL}}

\renewcommand{\epsilon}{\varepsilon}
\newcommand{\poly}{\operatorname{poly}}

\title{A Tight Lower Bound for 3-Coloring \\ Grids in the Online-LOCAL Model}

\author{\hspace{1cm} Yi-Jun Chang\footnote{National University of Singapore. Email: cyijun@nus.edu.sg} \and Gopinath Mishra\footnote{National University of Singapore. Email: gopinath@nus.edu.sg} \and Hung Thuan Nguyen\footnote{National University of Singapore. Email: hung@u.nus.edu}\hspace{1cm} \and Mingyang Yang\footnote{National University of Singapore. Email: e0589912@u.nus.edu} \and Yu-Cheng Yeh\footnote{National University of Singapore. Email: yu-cheng@comp.nus.edu.sg}}

\date{}

\begin{document}

\maketitle
\thispagestyle{empty}

\begin{abstract}
  Recently, \citeauthor*{akbari2021locality}~(ICALP 2023) studied the locality of graph problems in distributed, sequential, dynamic, and online settings from a {unified} point of view. They designed a novel $O(\log n)$-locality deterministic algorithm for proper 3-coloring bipartite graphs in the $\mathsf{Online}$-$\mathsf{LOCAL}$ model. In this work, we establish the optimality of the algorithm by showing a \textit{tight} deterministic $\Omega(\log n)$ locality lower bound, which holds even on grids. To complement this result, we have the following additional results:

\begin{enumerate}
    \item We show a higher and {tight} $\Omega(\sqrt{n})$ lower bound for 3-coloring toroidal and cylindrical grids.
    \item  Considering the generalization of $3$-coloring bipartite graphs to $(k+1)$-coloring $k$-partite graphs, 
    we show that the problem also has $O(\log n)$ locality when the input is a $k$-partite graph that admits a \emph{locally inferable unique coloring}. This special class of $k$-partite graphs covers several fundamental graph classes such as $k$-trees and triangular grids. Moreover, for this special class of graphs, we show a {tight} $\Omega(\log n)$ locality lower bound.
    \item For general $k$-partite graphs with $k \geq 3$, we prove that the problem of $(2k-2)$-coloring $k$-partite graphs exhibits a locality of $\Omega(n)$ in the $\onlineLOCAL$ model, matching the round complexity of the same problem in the $\LOCAL$ model recently shown by \citeauthor*{coiteux2023no}~(STOC 2024). Consequently, the problem of $(k+1)$-coloring $k$-partite graphs admits a locality lower bound of $\Omega(n)$ when $k\geq 3$, contrasting sharply with the $\Theta(\log n)$ locality for the case of $k=2$. 
\end{enumerate}
\end{abstract}

\newpage
\bigskip
\tableofcontents
\bigskip
\thispagestyle{empty}



\newpage
\pagenumbering{arabic}
\section{Introduction}
We focus on the \emph{locality} of graph problems, which has been studied in a variety of models. Throughout the paper, we only consider the \emph{deterministic} setting.

\begin{description}
    \item[Distributed setting:] In the classical $\LOCAL$ model~\cite{linial1992locality, peleg2000distributed} of distributed computing, an algorithm with locality $T$ processes the nodes of a graph simultaneously in parallel, in a way that each node determines its output by examining its neighborhood with a radius of $T$. Intuitively, an algorithm with locality $T$ in the $\LOCAL$ model can be run on a distributed network in $T$ synchronous communication rounds.
    \item[Sequential setting:] In the $\SLOCAL$ model~\cite{ghaffari2017complexity}, an algorithm with locality $T$ processes the nodes in a sequential order that is selected by an adversary. The output of a node may depend on its $T$-radius neighborhood and the outputs of the previously processed nodes. For instance, the well-known \emph{greedy coloring algorithm} solves the $(\text{degree} + 1)$-list coloring problem with locality $1$ in $\SLOCAL$.
    \item[Dynamic setting:] There have been several papers studying local algorithms in the dynamic setting~\cite{assadi2018fully, barenboim2019fully, bhattacharya2018dynamic, dobrev2013independent, gupta2018simple, ivkovic1993fully, neiman2015simple}, where an adversary constructs the graph \emph{dynamically}, adding or removing nodes and edges sequentially. Following each modification, an algorithm with locality $T$ is limited to adjusting the solution within the $T$-radius neighborhood of the point of change. Recently, \citet{akbari2021locality} formalized this setting by defining the two models $\dynamicLOCAL$ and $\dynamicLOCALpm$ to capture the incremental dynamic setting and the fully dynamic setting, respectively.
\end{description}

\citet{ghaffari2017complexity} studied the connections between $\LOCAL$ and $\SLOCAL$. They developed a method of simulating an arbitrary $\SLOCAL$ algorithm in the $\LOCAL$ model using \emph{network decompositions}. Combining this method with the deterministic polylogarithmic-round network decomposition of \citet{RozhonG20}, one can infer that the class of graph problems solvable with polylogarithmic locality deterministically is \emph{identical} in $\LOCAL$ and $\SLOCAL$.

Very recently, \citet{akbari2021locality} studied the locality of graph problems in distributed, sequential, and dynamic settings from a \emph{unified} point of view. They considered a new model $\onlineLOCAL$, which is a variant of $\SLOCAL$ that has a \emph{global memory}. Among all the models $\{\LOCAL$, $\SLOCAL$, $\dynamicLOCAL$, $\dynamicLOCALpm$, $\onlineLOCAL\}$, the $\onlineLOCAL$ model is the \emph{strongest} one: any algorithm in any of these models can be simulated in $\onlineLOCAL$ model with the same asymptotic locality. The $\LOCAL$ model is the \emph{weakest} model in the sense that
any $\LOCAL$ algorithm can be simulated in any of the above models with the same asymptotic locality.
Therefore, all the mentioned models are sandwiched between the $\LOCAL$ and $\onlineLOCAL$ models. Consequently, if one can match a locality lower bound in $\onlineLOCAL$ with a locality upper bound in $\LOCAL$, then it immediately implies a tight locality bound in all of the models.

\citet{akbari2021locality} obtained such a tight locality bound for a wide range of problems. For all \emph{locally checkable labeling} problems in paths, cycles, and rooted regular trees, they have nearly the same locality in all of the models:
\[
    \mathsf{LOCAL} \approx \mathsf{SLOCAL} \approx \mathsf{Dynamic}\text{-}\mathsf{LOCAL}^\pm \approx \mathsf{Dynamic}\text{-}\mathsf{LOCAL} \approx \mathsf{Online}\text{-}\mathsf{LOCAL}.
\]

\paragraph{Coloring bipartite graphs in $\onlineLOCAL$.}

While the above result suggests that these models can be quite similar, \citet{akbari2021locality} demonstrated an \emph{exponential} separation between the $\LOCAL$ and $\onlineLOCAL$ models. They designed a novel $O(\log n)$-locality algorithm for proper 3-coloring bipartite graphs in the $\onlineLOCAL$ model. In contrast, the same problem is known to have locality $\Omega(\sqrt{n})$ in the $\LOCAL$ model~\cite{brandt2017lcl}. Their work left open the following question.

\medskip
\centerline{%
    \parbox{0.7\linewidth}{
        \begin{mdframed}[hidealllines=true,backgroundcolor=gray!25]\begin{center}
                \emph{``Is it possible to find a 3-coloring in bipartite graphs in the $\onlineLOCAL$ model with locality
                    $o(\log n)$?''}
            \end{center}\end{mdframed}
    }%
}
\medskip

In this work, we resolve the question by demonstrating a \emph{tight} $\Omega(\log n)$ locality lower bound in the $\onlineLOCAL$ which holds even on grids, which are bipartite graphs.

\begin{restatable}{theorem}{main}
    In the $\onlineLOCAL$ model, the locality of 3-coloring a $\left(\sqrt{n} \times \sqrt{n}\right)$ grid is $\Omega(\log n)$.
    \label{thm:complexity}
\end{restatable}

Combining the above result with the $O(\log n)$ upper bound by \citet{akbari2021locality}, a tight bound is obtained.

\begin{corollary}\label{coro:main}
    In the $\onlineLOCAL$ model, the locality of 3-coloring bipartite graphs is $\Theta(\log n)$.
\end{corollary}

To establish \Cref{thm:complexity}, a main technical barrier that we overcome is that in $\onlineLOCAL$, we cannot rely on an \emph{indistinguishability} argument due to the presence of global memory. Most of the existing $\LOCAL$ lower bound proofs rely on such kind of an argument.
In particular, the known $\Omega(\sqrt{n})$ lower bound for 3-coloring grids in the $\LOCAL$ model~\cite{brandt2017lcl} works by first proving the lower bound on \emph{toroidal} grids and then using an indistinguishability argument to extend the lower bound to grids. The lower bound argument in~\cite{brandt2017lcl} heavily depends on the assumption that the underlying graph is a toroidal grid with an \emph{odd} number of columns, which is \emph{not} bipartite.

The presence of a global memory in $\onlineLOCAL$ allows the possibility for a graph problem to have different localities in grids and toroidal grids. Indeed, in this work, we show a much higher $\Omega(\sqrt{n})$ lower bound for 3-coloring on toroidal grids in the $\onlineLOCAL$ model. The lower bound does not contradict the $O(\log n)$-locality algorithm of~\cite{akbari2021locality} because toroidal grids are not bipartite in general.

\begin{restatable}{theorem}{torus}
    In the $\onlineLOCAL$ model, the locality of 3-coloring $\left(\sqrt{n} \times \sqrt{n}\right)$ toroidal and cylindrical grids is $\Omega(\sqrt{n})$.
    \label{thm:toroidal}
\end{restatable}

Combining the above $\onlineLOCAL$ lower bound with the trivial $O(\sqrt{n})$ locality upper bound in the $\LOCAL$ model, we establish that $\Theta(\sqrt{n})$ is a \emph{tight} locality bound for 3-coloring toroidal and cylindrical grids in all the models discussed above, as they are sandwiched between $\LOCAL$ and $\onlineLOCAL$.

\begin{corollary}\label{cor:lb}
    The locality of 3-coloring $\left(\sqrt{n} \times \sqrt{n}\right)$ toroidal and cylindrical grids is $\Theta(\sqrt{n})$ in all the models $\{\LOCAL$, $\SLOCAL$, $\dynamicLOCAL$, $\dynamicLOCALpm$, $\onlineLOCAL\}$.
\end{corollary}

Interestingly, our results yield two alternative proofs for the $\Omega(\sqrt{n})$ locality lower bound for 3-coloring grids in $\LOCAL$~\cite{brandt2017lcl}, which we briefly explain as follows. By an indistinguishability argument, the $\Omega(\sqrt{n})$ locality lower bound in \Cref{cor:lb} on toroidal and cylindrical grids in $\LOCAL$ automatically implies the same lower bound on grids in $\LOCAL$. Alternatively, the $\Omega(\log n)$ lower bound of \Cref{thm:complexity} in $\onlineLOCAL$ implies the same lower bound in $\LOCAL$, and we can automatically improve such a lower bound to $\Omega(\sqrt{n})$ by the known $\omega(\log^\ast n)$--$o(\sqrt{n})$ complexity gap on grids in $\LOCAL$~\cite{ChangKP16,ChangP17}.

\paragraph{Generalization: $(k + 1)$-coloring $k$-partite graphs in $\onlineLOCAL$.}
Next, we consider a natural generalization of the $3$-coloring problem in bipartite graphs to the task of $(k+1)$-coloring of $k$-partite graphs for any constant $k \geq 2$ within the $\onlineLOCAL$ model. 
Given \Cref{coro:main}, a pertinent question arises:
\begin{itemize}
    \item Is the locality of $(k+1)$-coloring $k$-partite graphs in $\onlineLOCAL$  $\Theta(\log n)$ for all $k$?
\end{itemize}
Interestingly, for $k \geq 3$, we can establish a lower bound of $\Omega(n)$. More specifically, we present a \emph{much stronger} result demonstrating an $\Omega(n)$ lower bound for $(2k-2)$-coloring.

\begin{restatable}{theorem}{higher}
    \label{thm:2k-2}
    Let $k \geq 2$ be a constant. The locality of $(2k-2)$-coloring a $k$-partite graph is $\Omega(n)$ in the $\onlineLOCAL$ model.
\end{restatable}
Observe that $k+1 \leq 2k-2$ for $k \geq 3$, so we have the following corollary.
\begin{corollary}\label{coro:higher}
    Let $k \geq 3$ be a constant. The locality of $(k+1)$-coloring a $k$-partite graph is $\Omega(n)$ in the $\onlineLOCAL$ model.
\end{corollary}

Recall that for the case of $k=2$, the locality $(k+1)$-coloring a $k$-partite graph is already known to be $\Theta(\log n)$ by \Cref{coro:main}.

\paragraph{Related work.} In the $\LOCAL$ model, the same lower bound of $\Omega(n)$ for $(2k-2)$-coloring $k$-partite graphs has been very recently established in \cite{coiteux2023no}. More generally, they showed that the round complexity of the problem of $c$-coloring $k$-partite graphs is $\widetilde{\Theta}(n^{1/\alpha})$ in $\LOCAL$, where $\alpha=\lfloor \frac{c-1}{k-1} \rfloor$ and the notation $\widetilde{\Theta}(\cdot)$ hides polylogarithmic factors. Notably, this bound is polynomial in $n$ for any constants $k$ and $c$. 

For the case of $(c,k)=(3,2)$, there is a huge gap between the complexities of $3$-coloring bipartite graphs in $\LOCAL$ ($\widetilde{\Theta}(\sqrt{n})$, \cite{coiteux2023no}) and $\onlineLOCAL$ ($\Theta(\log n)$, \Cref{coro:main}). On the other hand, for $k \geq 3$ and $c \leq 2k-2$, \Cref{thm:2k-2} says that the locality of $c$-coloring $k$-partite graphs is $\Omega(n)$ in  $\onlineLOCAL$ -- a match with the round complexity of the same problem in $\LOCAL$.


\paragraph{Improved upper bounds for special graph classes.}
For the problem of $(k+1)$-coloring of $k$-partite graphs in $\onlineLOCAL$, the locality is $\Theta(\log n)$ for $k=2$ (\Cref{coro:main}) and is  $\Omega(n)$ for $k\geq 3$ (\Cref{coro:higher}). A plausible explanation for this notable shift in the locality bound when transitioning from $k=2$ to $k\geq 3$ is as follows: While the bipartition (i.e., 2-coloring) of a (connected) bipartite graph is unique, a $k$-partite graph might not admit a unique $k$-coloring in general. Intuitively, this means that for $k\geq 3$, the colorings of two graph fragments computed by an $\onlineLOCAL$ algorithm may be completely incompatible, requiring more adjustments to combine them, thus leading to a locality of $\Omega(n)$.

We show that it is possible to break the $\Omega(n)$ locality lower bound of \Cref{thm:2k-2} if the underlying $k$-partite graph admits a unique $k$-coloring satisfying the following desirable property: For each graph fragment seen by an $\onlineLOCAL$ algorithm, the unique $k$-coloring restricted to that fragment can be inferred correctly, up to a permutation of the $k$ colors. As we will later see, examples of such graphs include all (connected) \emph{bipartite graphs}, \emph{triangular grids}, and \emph{$k$-trees}.




Let $\ell\in \mathbb{N}$ be a non-negative integer.
Let $G$ be a $k$-partite graph.
Let $G' = (V', E')$ be a connected subgraph of $G$.
Define $\mathcal{C}^\ast(G', \ell)$ as the set of all proper $k$-colorings of the subgraph of $G$ induced by the  $\ell$-radius neighborhood $\mathcal{B}(V', \ell)$ of $V'$.
Define $\mathcal{C}(G', \ell)$ as the set of all proper $k$-colorings of $G'$ resulting from restricting some coloring $c \in \mathcal{C}^\ast(G', \ell)$ to the domain $V'$. We say that two $k$-colorings $c$ and $c'$ are identical \emph{up to a permutation} if there exists a permutation $\phi \colon [k] \to [k]$ such that $c' = \phi \circ c$. 

\begin{definition}[Locally inferable unique colorings]
    \label{def:lpcc}
    We say that a $k$-partite graph $G$ has a locally inferable unique $k$-coloring with radius $\ell$ if all colorings in $\mathcal{C}(G', \ell)$ are identical up to a permutation for any connected subgraph $G'$ of $G$.
\end{definition}

To put it another way, $G$ admits a locally inferable unique $k$-coloring with radius $\ell$ if the $k$-color partition of $G$ is unique and can be inferred locally in the following way: For any connected subgraph $G'$ of $G$, the unique $k$-color partition restricted to $G'$ can be inferred by inspecting its $\ell$-radius neighborhood. 
Since the bipartition of a bipartite graph is unique, every bipartite graph has a locally inferable unique $k$-coloring with radius $\ell = 0$. We show that triangular grids and $k$-trees also admit locally inferable unique colorings with $\ell \in O(1)$.

\begin{description}
    \item[Triangular grids:] The triangular grid $G=(V,E)$ of side length $d$ is defined as follows.
    \begin{itemize}
        \item $V = \left\{ (x,y)\in \mathbb{N}^2 \, : \,  0 \leq x+y \leq d \right\}$.
        \item For any two nodes $(x,y)$ and $(x',y')$ in $V$, add an edge between them to $E$ if $|x-x'| + |y-y'| = 1$ or $x-x' = y-y' \in \{-1, 1\}$.
    \end{itemize}
    Consider a connected subgraph $G'=(V',E')$ of a triangular grid $G$. See \Cref{fig:triangular_oracle} for an illustration. For any node $u \in V'$, there is a triangle $\Delta_u$ containing $u$ in the $1$-radius neighborhood of $V'$ in $G$.
    Observe that for any two nodes $u$ and $v$ in $V'$, they are connected by a sequence of triangles $(\Delta_u = \Delta_1, \Delta_2, \ldots, \Delta_t = \Delta_v)$ in the $1$-radius neighborhood of $V'$ in $G$: For each $i \in [t - 1]$, $\Delta_i$ and $\Delta_{i + 1}$ shares an edge. If the nodes of $\Delta_i$ and $\Delta_{i + 1}$ are $\{a, b, c\}$ and $\{b, c, d\}$, respectively, then the color of $a$ and $d$ must be identical under any $3$-coloring. In particular, fixing the coloring of $\Delta_u$ uniquely determines the color of $v$. Therefore, for any given $3$-coloring of the $1$-radius neighborhood of $V'$ in $G$, its restriction to $V'$ must be unique up to a permutation. In other words, $G$ admits a locally inferable unique $3$-coloring with radius $\ell=1$.
  
    \item[$k$-trees:] The graphs that can be constructed as follows are called $k$-trees: Start with a $(k + 1)$-clique, and then iteratively add nodes such that each newly added node $v$ is adjacent to exactly $k$ existing nodes that form a $k$-clique. The construction of a $k$-tree $G$ naturally yields a tree $H$ so that each node $h \in V(H)$ corresponds to a $(k + 1)$-clique of $G$ and two nodes of $H$ are adjacent if their corresponding cliques share $k$ nodes. Consider a connected subgraph $G'=(V',E')$ of a $k$-tree $G$. Similar to the case of triangular grids, by querying the $1$-radius neighborhood of $V'$ in $G$, we learn all the $(k + 1)$-cliques containing at least one node in $V'$ and their adjacency relations in $H$. Therefore, $k$-trees admit a locally inferable unique $(k+1)$-coloring with radius $\ell=1$, as fixing the color assignment of any one of these $(k + 1)$-cliques fixes the color of all other nodes in $G'$.
\end{description}

\begin{figure}
    \centering
    \includegraphics[width=0.3\linewidth]{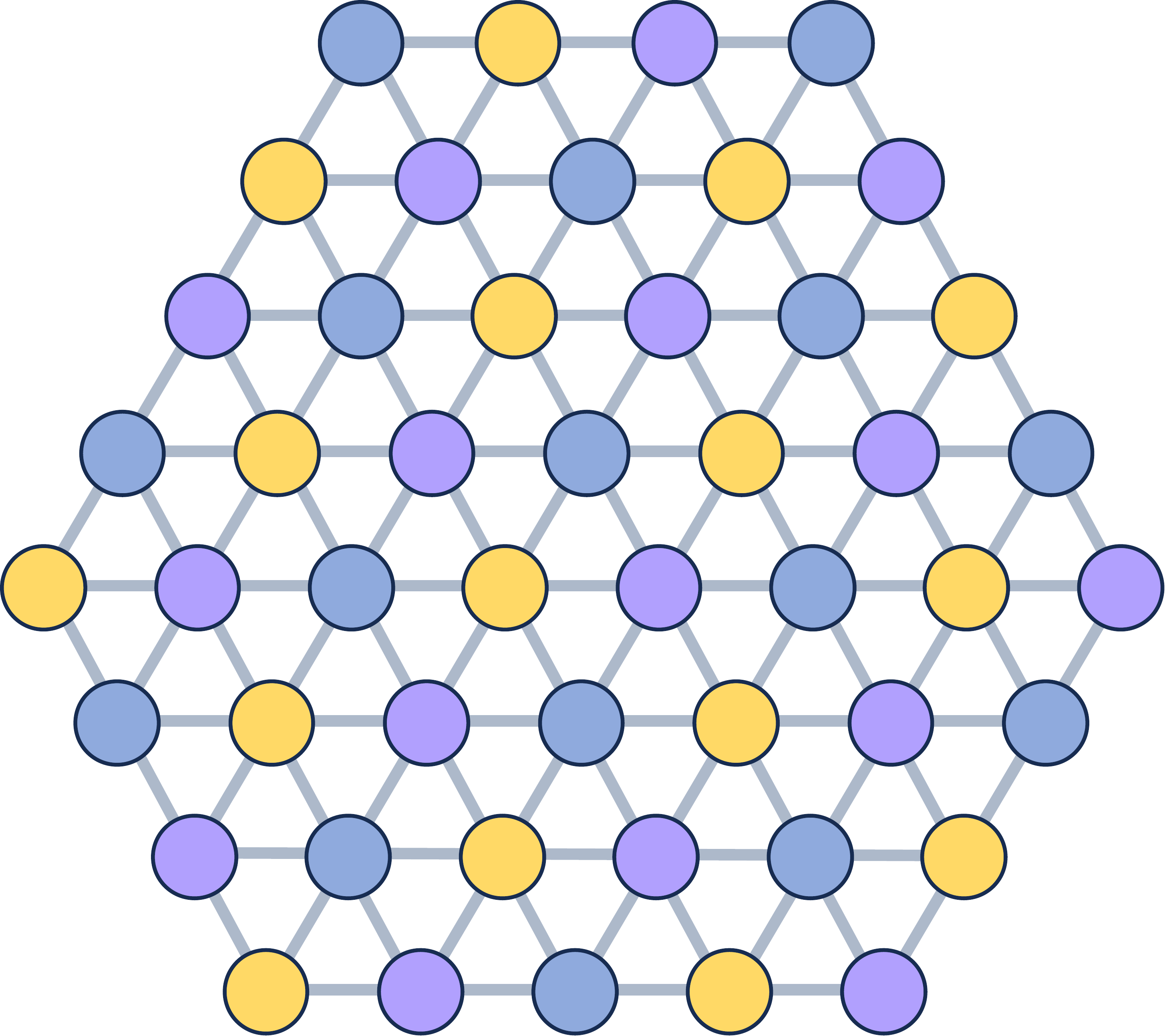}
    \caption{The unique tripartition of a triangular grid restricted to a connected subgraph.}
    \label{fig:triangular_oracle}
\end{figure}

For notational simplicity, from now on, we write $\mathcal{L}_{k,\ell}$ to denote the set of all $k$-partite graphs having a locally inferable unique $k$-coloring with radius $\ell$. 
We show that a $(k+1)$-coloring for any graph in $\mathcal{L}_{k,\ell}$ with $k \in O(1)$ and $\ell \in O(1)$ can be computed with $O(\log n)$ locality in the $\onlineLOCAL$ model.

\begin{restatable}{theorem}{lpccupper}
    \label{thm:lpccupper}
    Let $k \geq 2$ be any constant. There is an $\onlineLOCAL$ algorithm with locality $O(\log n)$ that $(k+1)$-colors any graph in $\mathcal{L}_{k,\ell}$ with $\ell \in O(1)$.
\end{restatable}

The following corollary is an immediate consequence of \Cref{thm:lpccupper}.

\begin{corollary} 
    The locality of 3-coloring triangular grids and $(k+2)$-coloring $k$-trees with $k \in O(1)$ is $O(\log n)$ in the $\onlineLOCAL$ model.
\end{corollary}


We complement \Cref{thm:lpccupper} with a matching lower bound by extending the lower bound of \Cref{thm:complexity} from $k=2$ to higher values of $k$.

\begin{restatable}{theorem}{lpcclower}
    \label{thm:lpcclower}
 Let $k \geq 2$ be any constant. For any $\onlineLOCAL$ algorithm that $(k+1)$-colors any graph in $\mathcal{L}_{k,\ell}$ with $\ell \in O(1)$, its locality must be $\Omega(\log n)$.
\end{restatable}

Combining \Cref{thm:lpcclower,thm:lpccupper}, we obtain a tight bound for $(k + 1)$-coloring for $\mathcal{L}_{k,\ell}$ with $\ell \in O(1)$.

\begin{restatable}{corollary}{lpctight}
    For any constant $k \geq 2$, in the $\onlineLOCAL$ model, the locality of $(k+1)$-coloring graphs in $\mathcal{L}_{k,\ell}$  with $\ell \in O(1)$ is $\Theta(\log n)$.
\end{restatable}

\paragraph{Related models.} Subsequent to the initial publication of this paper, \citet{akbari2024online} studied the \emph{randomized} variant of the $\onlineLOCAL$ model, extending our $\Omega(\log n)$ lower bound to the randomized setting. In addition to the $\onlineLOCAL$ model, researchers have studied other hybrids of shared-memory and message-passing models, including the \emph{message-and-memory} (m\&m) model~\cite{aguilera2018passing} and the \emph{computing-with-the-cloud} (CWC) model~\cite{afek2024distributed}.

\subsection{Technical Overview}

We give an overview of the proofs of \Cref{thm:complexity,thm:toroidal,thm:2k-2,thm:lpccupper,thm:lpcclower}.

\subsubsection{Lower Bound for 3-Coloring Bipartite Graphs}

We start by reviewing the $\Omega(\sqrt{n})$ locality lower bound for 3-coloring grids in $\LOCAL$~\cite{brandt2017lcl}. First of all, by an indistinguishability argument, they can focus on toroidal grids and not grids. Without loss of generality, they assumed that the coloring has a special property that any node colored with 3 is adjacent to a node colored with 1 and a node colored with 2. They showed that for any proper 3-coloring of a toroidal grid satisfying the above condition, it is possible to orient the diagonals between any two nodes colored with 3 in such a way that yields an Eulerian graph over the set of all nodes colored with 3.

They made a key observation that the difference between the number of times a directed edge in the Eulerian graph passes a row in the upward direction and in the downward direction is \emph{invariant} of the choice of the row. Let $s(\ell)$ denote this value, where $\ell$ is the number of nodes in a row. They showed that when $\ell$ is odd, $s(\ell)$ must be an odd number within the range $[-\ell/2, \ell/2]$.  Based on this property, they obtained the desired lower bound by a reduction to the following problem on an $n$-cycle, which requires $\Omega(n)$ rounds to solve: Each node outputs a number in such a way that all the numbers sum up to exactly $q$, where $q$ is any given odd integer in the range $[-n/2, n/2]$.

The above lower bound proof does \emph{not} apply to $\onlineLOCAL$ because it highly depends on the assumption that the underlying graph is a toroidal grid, which is fine in the $\LOCAL$ model due to an indistinguishability argument. This is not possible for $\onlineLOCAL$ due to the presence of a global memory. To extend the lower bound to $\onlineLOCAL$, new ideas are required.

Our $\onlineLOCAL$ lower bound starts from the following informal observation. The notion of $s(\ell)$ can be extended beyond rows. Informally, for any directed path or directed cycle, we may define its $s$-value as the difference between the number of times a directed edge in the Eulerian graph passes it in the rightward direction and the leftward direction. Similarly, it is possible to argue that if two directed paths or two directed cycles are \emph{homotopic} to each other, then their $s$-values must be identical. Moreover, this observation still makes sense even if the underlying graph is a grid and not a toroidal grid, although the directed graph defined on the set of all nodes colored with 3 might not be Eulerian in that case.

For the case of grids, one can observe that the $s$-value of any directed cycle must be zero. This gives rise to the following proof idea in the $\onlineLOCAL$ model. Suppose we have an adversary strategy that can force an algorithm to create a directed path in a row on the grid that has an $s$-value of at least $t$. In that case, such an algorithm must have locality $\Omega(t)$, since otherwise using such a directed path construction, the adversary can force a cycle to have an $s$-value that is nonzero, which is a contradiction.

A technical difficulty in realizing the above proof idea is that in the informal definition of the $s$-value of a directed path or directed cycle suggested above, the $s$-value not only depends on the colors on the path or cycle under consideration but also depends on the colors of all nodes adjacent to the path or cycle. Moreover, there are subtle issues about how we deal with an intersection that occurs at the endpoints of a path. This makes the informal definition of $s$-value inconvenient to use, especially because we have to consider partial coloring to prove lower bounds in the $\onlineLOCAL$ model.

To deal with the above issue, we develop the notion of $b$-value that depends \emph{only} on the coloring of the considered path or cycle. The notion of $b$-value  can be seen as a simplification of the notion of $s$-value that still enjoys all the nice properties of $s$-value that we need. It is noteworthy that our notion of $b$-value does not require the assumption that any node colored with 3 is adjacent to a node colored with 1 and a node colored with 2.

To realize the above proof idea, we demonstrate a recursive strategy that constructs a directed path with length at most $T \cdot c^k$ whose $b$-value is at least $k$, for some constant $c > 0$ and any given $T$-locality $\onlineLOCAL$ algorithm. By setting both $T$ and $k$ to be $\Theta(\log n)$ with a sufficiently small leading constant, the path length can be made smaller than $\sqrt{n}$, so we can fit the path in a $\left(\sqrt{n} \times \sqrt{n}\right)$ grid. This gives us the desired $\Omega(\log n)$ locality lower bound. A key idea in the recursive construction is to utilize a property of $b$-value that whether it is even or odd is determined by the parity of the path length and the colors of the two endpoints.

To show the higher lower bound of $\Omega(\sqrt{n})$ in cylindrical and toroidal grids, we again utilize the property that the $b$-value of each row is invariant of the choice of the row and is odd given that each row contains an odd number of nodes. Therefore, as an adversary, we may just ask the algorithm to produce a coloring of two rows, and then we set their directions in a way to force them to have different $b$-values.

\subsubsection{Lower Bound for \texorpdfstring{$(2k-2)$}{(2k - 2)}-Coloring \texorpdfstring{$k$}{k}-Partite Graphs}
To establish the $\Omega(n)$ lower bound for $(2k-2)$-coloring $k$-partite graphs in the $\LOCAL$ model, \cite{coiteux2023no} considered an instance $G^\#$ such that the chromatic number of $G^\#$ is $2k-1$, while the subgraphs induced by the $o(n)$-radius neighborhood of any node of $G^\#$  are $k$-partite. It is crucial to note that $G^{\#}$ lies outside the input family. Despite being outside the input family, the instance is locally solvable, signifying that the subgraph induced by the $o(n)$-neighborhood of any node can be colored with $k$ colors. The lower bound in the $\LOCAL$ model is established through an indistinguishability argument by \citet{linial1992locality}, which proves lower bounds by employing hard instances taken from outside the input family:  To show a $T$-round $\LOCAL$ lower bound for $c$-coloring for $G$ by this approach, one just needs to construct a graph $G^\#$ that cannot be $c$-colored and is indistinguishable to $G$ by a $T$-round $\LOCAL$ algorithm.

The construction of the hard instance $k$-partite graph $G^*$, to establish the lower bound for $(2k-2)$-coloring $k$-partite graphs in the $\onlineLOCAL$ model, is inspired by the construction of $G^{\#}$. However, the application of Linial's technique, designed for $\LOCAL$, is not straightforward in the $\onlineLOCAL$ model possibly due to the presence of global memory. Despite the inspiration of our hard instance being inspired from the hard instance in the $\LOCAL$ model, our arguments to establish the lower bound in the $\onlineLOCAL$ model are different, as sketched in the subsequent paragraphs, and accommodate the unique features of the $\onlineLOCAL$ setting, including the utilization of global memory.

In the construction of the hard instance graph $G^*$, we consider $n/k^2$ gadgets, each comprising $k^2$ nodes. The node set of a gadget is represented as elements in a $k \times k$ matrix. Two nodes in a gadget are connected by an edge if and only if they are neither in the same row nor the same column. Under a proper coloring of a gadget, the gadget is said to be \emph{row-colorful} (\emph{column-colorful}) if there exists a row (column) with all the $k$ nodes being colored by distinct colors. We show that a gadget is exactly one out of row-colorful and column-colorful when we have a proper coloring of the gadget with $2k-2$ colors. $G^*$ is formed by assembling $n/k^2$ gadgets, connecting nodes in consecutive gadgets as follows: for any two consecutive gadgets, two nodes (one from each) are connected by an edge if and only if they are neither in the same row nor the same column. We prove that, in any proper coloring of $G^*$ with $2k-2$ colors, either all the gadgets are row-colorful or all of them are column-colorful.

The basis of our lower bound argument is as follows: if the locality of the algorithm in the $\onlineLOCAL$ model is $o(n)$, then the adversary can ask the algorithm to color $G^*$ in a specific order such that the first gadget is row-colorful and the last gadget is column-colorful. This scenario results in the impossibility of properly coloring some nodes, as all the gadgets must be either row-colorful or column-colorful.

In \cite{coiteux2023no}, the authors consider the problem of $c$-coloring $k$-partite graphs, where $c$ and $k$ are constants. This is a more general version than both $3$-coloring bipartite graphs and $(2k-2)$-coloring $k$-partite graphs. Specifically, \cite{coiteux2023no} establishes that the round complexity of this generalized problem is essentially $\widetilde{\Theta}(n^{1/\alpha})$, where $\alpha=\lfloor \frac{c-1}{k-1} \rfloor$. We have discussed the extension of the hard instance construction for $c=2k-2$ (and thus $c\leq 2k-2$) to demonstrate an $\Omega(n)$ lower bound for $(2k-2)$-coloring $k$-partite graphs in the $\onlineLOCAL$ model. It is natural to consider the possibility of such an extension for all $c$ and $k$. However, if such a general extension is possible for all $c$ and $k$, we would have achieved $\Omega(\sqrt{n})$ lower bound on the locality of $3$-coloring bipartite graphs in the $\onlineLOCAL$ model, which is impossible as the locality of $3$-coloring bipartite graphs is $\Theta(\log n)$ due to \Cref{coro:main}. Resolving the locality of $c$-coloring $k$-partite graphs for all values of $c$ and $k$ in the $\onlineLOCAL$ model remains an intriguing open problem.

\subsubsection{Coloring Graphs with Locally Inferable Unique Colorings}\label{sec:overview-perm}

From the upper bound side, \citet{akbari2021locality} presented an $O(\log n)$-locality $\onlineLOCAL$ algorithm of $3$-coloring a bipartite graph. Their algorithm is based on a key observation that the bipartition for any (connected) bipartite graph $G$ is unique, so there are exactly two different ways to 2-color $G$. In the subsequent discussion, we refer to these two possibilities as two \emph{parities}.

The algorithm of \citet{akbari2021locality}, in a nutshell, tries to 2-color the graph fragments seen by the $\onlineLOCAL$ algorithm using the colors in $\{1, 2\}$. A technical challenge arises when two graph fragments with different parities become connected, as their 2-coloring are incompatible with each other. The issue is resolved by \emph{flipping} the parity of one graph fragment by adding a layer of the third color on the boundary. The flipping operation requires $O(1)$ locality. 
By always flipping the graph fragment with a smaller size, it is guaranteed that at most $O(\log n)$ flipping operations are needed from the perspective of each node. Therefore, the total required locality is $O(\log n)$.

In this work, we extend their algorithm to $(k+1)$ coloring all $k$-partite graphs having a locally inferable unique coloring with radius $\ell \in O(1)$. For these $k$-partite graphs, there is a unique partition into $k$ parts that can be inferred locally: Given a connected node set $S$, the partition restricted to $S$ can be inferred by information within the $\ell$-radius neighborhood of $S$. Therefore, we can assume that there is an oracle $\mathcal{O}$ that returns such a $k$-partition for any graph fragment seen by the $\onlineLOCAL$ algorithm, as the cost of implementing such an oracle is just an extra $\ell \in O(1)$ locality. Our $(k+1)$ coloring algorithm is based on the generalization of the notion of parity into the assignment of the \emph{indices} $\{1,2, \ldots, k\}$ to the $k$ parts in the partition, which is possible due to the oracle $\mathcal{O}$.
Similar to the algorithm of \citet{akbari2021locality}, we try to $k$-color the graph fragments seen by the $\onlineLOCAL$ algorithm using the colors in $[k]$. When two graph fragments with incompatible $k$-coloring become connected, we show that with $O(1)$ locality, it is possible to unify their coloring by iteratively swapping the indices with the help of the extra color $k+1$. As such, we can $(k + 1)$-color any graph in $\mathcal{L}_{k,\ell}$ for any constants $k$ and $\ell$ with $O(\log n)$ locality.

Now, we discuss the idea behind $\Omega(\log n)$ lower bound on the locality of $(k+1)$-coloring $k$-partite graphs when the input graph is from $\mathcal{L}_{k,\ell}$ with $\ell \in O(1)$. 
Recall that \Cref{thm:complexity} already establishes this result for the case of $k = 2$, as grids are bipartite graphs, which belong to $\mathcal{L}_{2,0}$. To extend the result to higher values of $k$, we construct a sequence of graphs $(G_2, G_3, \ldots, G_k)$ recursively such that $G_k \in \mathcal{L}_{k,\ell}$ with $\ell \in O(1)$. For the base case, $G_2$ is a $\left(\sqrt{n} \times \sqrt{n}\right)$ grid. For the inductive step, $G_{i+1}$ is constructed from $G_i$ by augmenting each node $u$ in $G_{i}$ with a new node $u'$ that is connected to $u$ and its (original) neighbors.  We show that an $o(\log n)$-locality algorithm that properly colors $G_{i}$ with $i+1$ colors can be obtained from an $o(\log n)$-locality algorithm that properly colors $G_{i+1}$ with $i+2$ colors, so the  $\Omega(\log n)$ lower bound for the case of $k=2$ (\Cref{thm:complexity}) implies the same asymptotic lower bound for any constant $k \geq 2$.


\subsection{Roadmap}

In \Cref{sect:prelim}, we introduce essential graph terminology and define the computation models. 
In \Cref{sec:main-lower}, we present our main lower bound results: an $\Omega(\log n)$ locality lower bound for $3$-coloring in simple grids and a higher $\Omega(\sqrt{n})$ lower bound for toroidal and cylindrical grids (\Cref{thm:complexity,thm:toroidal}).
In \Cref{sec:2k-2}, we prove the $\Omega(n)$ locality lower bound for $(2k-2)$-coloring of $k$-partite graphs (\Cref{thm:2k-2}). 
In \Cref{sec:lpcc}, we establish a tight $\Theta(\log n)$ bound for $(k + 1)$-coloring in locally inferable unique coloring graphs (\Cref{thm:lpccupper,thm:lpcclower}).

\section{Preliminaries}\label{sect:prelim}

Throughout this paper, the input graph $G=(V, E)$ is simple, connected, undirected, and finite unless explicitly specified otherwise. We write $n = |V|$ to denote the number of nodes in $G$. For a subset of nodes $U \subseteq V$ and a natural number $T$, we write $\mathcal{B}(U, T)$ to denote the set of all nodes within the $T$-radius neighborhood of any $v \in U$. If there is only one node $v$ in $U$, we may write $\mathcal{B}(v, T)$ instead of $\mathcal{B}(U, T)$. We write $G[U]$ to denote the subgraph of $G$ induced by the set of nodes $U \subseteq V$. In particular, $G[\mathcal{B}(v, T)]$ is the subgraph of $G$ induced by the $T$-radius neighborhood of $v$. For any graph $H$, we write $V(H)$ and $E(H)$ to denote its node set and edge set, respectively. For a subgraph $H \subseteq G$ and a node $u$, we write $H - u$ to denote the induced subgraph $H[V(H) \setminus \{u\}]$. For $x \in \mathbb{N}$, we write $[x]$ to denote $\{1,2, \ldots, x\}$.

\subsection{Grids}

We investigate local algorithms in three types of grid topologies.

\begin{description}
    \item[Simple grids:] A simple grid of dimension $(a \times b)$ is a grid with $a$ rows and $b$ columns. More formally, an $(a \times b)$ grid is defined by the set of nodes $\{ (i,j) \colon i\in [a] \text{ and } j \in [b]\}$ such that two nodes $(i,j)$ and $(i', j')$ are adjacent if and only if $|i-i'| + |j-j'| = 1$. For each $i \in [a]$, the set of nodes $\{ (i,j) \colon j \in [b]\}$ forms a \emph{row}. For each $j \in [b]$, the set of nodes $\{ (i,j) \colon i \in [a]\}$ forms a \emph{column}. In a simple grid, each row and each column induce a path.
    \item[Cylindrical grids:] A cylindrical grid of dimension $(a \times b)$ is the result of adding the edges in $\{ \{(i,1), (i,b)\} \colon i \in [a]\}$ to an $(a \times b)$ grid. In other words, a cylindrical grid is the result of connecting the left border and the right border of a simple grid. In a cylindrical grid, each row induces a cycle and each column induces a path.
    \item[Toroidal grids:] A toroidal grid of dimension $(a \times b)$ is the result of adding the edges in $\{ \{(i,1), (i,b)\} \colon i \in [a]\}$ and $\{ \{(1, j), (a, j)\} \colon j \in [b]\}$ to an $(a \times b)$ grid. In other words, a toroidal grid is the result of connecting the left and right borders and connecting the upper and lower borders of a simple grid. In a toroidal grid, each row and each column induce a cycle.
\end{description}

In \Cref{thm:complexity}, we consider simple grids of dimension $\left(\sqrt{n} \times \sqrt{n}\right)$. In \Cref{thm:toroidal}, we consider toroidal and cylindrical grids of dimension $\left(\sqrt{n} \times \sqrt{n}\right)$. It is worth noting that toroidal grids and cylindrical grids with an odd number of columns are not bipartite graphs.

\subsection{Models}\label{sect:models}

We proceed to define two computation models: $\LOCAL$ and $\onlineLOCAL$. In these models, the complexity of an algorithm is measured by its \emph{locality}, which informally refers to the minimum number $T$ such that the output of a node is determined by its  $T$-radius neighborhood. Generally, the locality of an algorithm can be a function of the number $n = |V|$  of nodes in the input graph $G=(V, E)$. It is assumed that the algorithm knows the value of $n$. Throughout the paper, we only consider \emph{deterministic} algorithms, so it does not matter whether the adversary is oblivious or adaptive in the following definitions. In both $\LOCAL$ and $\onlineLOCAL$, the adversary assigns unique identifiers from the set $\{1, 2, \ldots, \poly(n)\}$ to the nodes.

\paragraph{The $\LOCAL$ model.} In an algorithm with locality $T$ in the $\LOCAL$ model~\cite{linial1992locality, peleg2000distributed}, each node independently determines its part of the output based on the information within its $T$-radius neighborhood. For example, for the proper $k$-coloring problem, each node $v \in V$ just needs to output its own color $c(v) \in [k]$ in such a way that no two adjacent nodes output the same color. The output of a node may depend on factors such as the graph structure, input labels, and unique identifiers.

There is an alternative way of defining the $\LOCAL$ model from the perspective of distributed computing by viewing the input graph $G=(V, E)$ as a communication network, where each node $v \in V$ is a machine and each edge $e \in E$ is a communication link. The communication proceeds in synchronous rounds. In each round, each node can communicate with its neighbors by exchanging messages of unlimited size. The locality of an algorithm is the number of communication rounds.

\paragraph{The $\onlineLOCAL$ model.} In an algorithm of locality $T$ in the $\onlineLOCAL$ model~\cite{akbari2021locality}, nodes are processed sequentially according to an adversarial input sequence $\sigma = (v_1, v_2, \ldots, v_n)$. Let $G_i = G\left[\bigcup_{j=1}^i \mathcal{B}(v_j, T)\right]$ denote the subgraph induced by the $T$-radius neighborhoods of the first $i$ nodes in $\sigma$. When the adversary presents a node $v_{i}$, the algorithm must assign an output label to $v_{i}$ based on the information from the subsequence $(v_1, v_2, \ldots, v_{i})$ and the induced subgraph $G_i$.

In the related $\SLOCAL$ model~\cite{ghaffari2017complexity}, the assignment of the output label of $v_{i}$ can \emph{only} be based on $G\left[\mathcal{B}(v_i, T)\right]$ and the output labels that were already assigned to the nodes in $\mathcal{B}(v_i, T)$, so $\onlineLOCAL$ can be seen as a variant of $\SLOCAL$ that has a global memory.

\section{Hardness of 3-Coloring in Grids}\label{sec:main-lower}

In this section, we establish our lower bound results for $3$-coloring in a simple grid, as well as toroidal and cylindrical grids. In \Cref{sec:properties}, we discuss some groundwork to prove the lower bounds presented in the subsequent sections. Moving on to \Cref{sec:lower}, we prove \Cref{thm:complexity}, showing the $\Omega(\log n)$ locality lower bound for $3$-coloring in a simple grid. Subsequently, in \Cref{sec:lower2}, we prove \Cref{thm:toroidal}, establishing the $\Omega(\sqrt{n})$ locality lower bound for $3$-coloring in toroidal and cylindrical grids.

\subsection{Properties of 3-Coloring in Grids}
\label{sec:properties}
In this section, we focus on analyzing the properties of proper 3-coloring and introducing the notion of $b$-value, which plays a critical role in the proof of \Cref{thm:complexity,thm:toroidal}. Although some of the results in this section apply to general 3-partite graphs, they are particularly interesting when we restrict our focus to grid topologies. Throughout this section, we let $c \colon V \rightarrow \{1, 2, 3\}$ be any proper 3-coloring of any 3-partite graph $G=(V,E)$.

\begin{definition}[$a$-value]
    Let $\{u,v\} \in E$ be an edge in $G$. Define the $a$-value of $(u,v)$ by
    \begin{equation*}
        a(u, v)=
        \begin{cases}
            c(u) - c(v) & \text{if $c(u) \neq 3$ and $c(v) \neq 3$,} \\
            0           & \text{otherwise.}
        \end{cases}
    \end{equation*}
\end{definition}

For any edge $\{u,v\}\in E$, note that $a(u,v)\in \{-1,0,1\}$. Moreover, the $a$-value of $(u,v)$ is non-zero if and only if one of $\{u,v\}$ is colored with $1$ and the other one is colored with $2$. Observe that we always have $a(u, v) + a(v, u) = 0$ for any edge $\{u,v\}\in E$.

We define the $b$-value of a directed path $P$ or a directed cycle $C$ by the sum of $a(u,v)$ over all directed edges $(u,v)$ in $P$ or $C$, respectively.

\begin{definition}[$b$-value]
    Let $H$ be a directed path or a directed cycle. Define the $b$-value of $H$ by
    \[
        b(H) = \sum_{(u,v) \in H} a(u, v).
    \]
\end{definition}

For the special case of a zero-length directed path $P$, we define $b(P) = 0$, as it does not contain any directed edge.

We observe that $b(C) = 0$ for any 4-node directed cycle $C$. This observation is particularly relevant to grid topologies, as they consist of ``cells'' that are 4-node cycles.

\begin{lemma}[Cancellation of the $a$-values within a cell]
    \label{lem:cell}
    For any 4-node directed cycle $C$, $b(C) = 0$.
\end{lemma}

\begin{proof}
    Let  $C = \begin{array}{ccc}
            t          & \leftarrow  & s        \\
            \downarrow &             & \uparrow \\
            u          & \rightarrow & v
        \end{array}$, so we have $b(C) = a(u,v) + a(v,s) + a(s,t) + a(t,u)$.
    Since there are only three colors, we have $c(u) = c(s)$ or $c(v) = c(t)$ or both. Without loss of generality, we assume $c(u) = c(s)$, so
    \[
        a(u,v) + a(v,s) = a(u,v) + a(v,u) = 0 \ \ \ \text{and} \ \ \ a(s, t) + a(t, u) = a(s,t) + a(t,s) = 0,
    \]
    implying that $b(C) = a(u,v) + a(v,s) + a(s,t) + a(t,u) = 0$.
\end{proof}

We say that a cycle is \emph{simple} if it does not have repeated nodes. As a consequence of \Cref{lem:cell}, we show that $b(C) = 0$ for any simple directed cycle $C$ in a grid.

\begin{lemma}[$b$-value of a cycle]
    For any simple directed cycle $C$ in a grid, $b(C) = 0$.
    \label{lem:cycle}
\end{lemma}

\begin{proof}
    We can calculate $b(C) = \sum_{(u, v) \in C} a(u, v)$ alternatively by summing up the $b$-value of each cell that is inside the cycle, where each cell is seen as a 4-node directed cycle that has the same orientation as that of $C$, so $b(C) = 0$ by \Cref{lem:cell}. The validity of the calculation is due to the cancellation of the $a$-values associated with the edges that are strictly inside $C$, as illustrated in \Cref{fig:grid}. For each directed edge $(u,v)$ on $C$, $a(u,v)$ appears exactly once in the calculation. For each edge $\{u,v\}$ that is inside $C$, both $a(u,v)$ and $a(v,u)$ appear exactly once in the calculation, so they cancel each other.
\end{proof}

\begin{figure}[ht]
    \centering
    \includegraphics[scale=0.55]{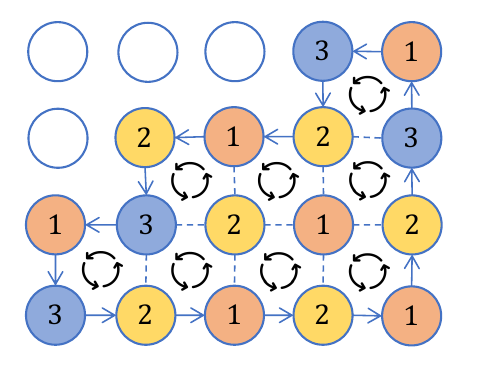}
    \caption{Cancellation of $a$-values.}
    \label{fig:grid}
\end{figure}

For each node $u \in V$, we define
\begin{equation*}
    i(u)=
    \begin{cases}
        1 & \text{if } c(u) = 3, \\
        0 & \text{otherwise.}
    \end{cases}
\end{equation*}
as the indicator variable for $c(u) = 3$. We show that the parity of the $b$-value of a directed path is uniquely determined by the colors of the two endpoints and the parity of the path length.

\begin{lemma}[Parity of $b$-value]\label{lem:parity}
    Let $C$ be any directed {cycle} of length $\ell$ and let $P=(u, \ldots, v)$ be any directed path of length $\ell$.
    We have
    \[b(C) \equiv  \ell \pmod{2}  \qquad \text{and} \qquad b(P) \equiv i(u) + i(v) + \ell \pmod{2}.\]
\end{lemma}
\begin{proof}
    We just need to focus on proving $b(P) \equiv i(u) + i(v) + \ell \pmod{2}$, as $b(C) \equiv  \ell \pmod{2}$ follows from $b(P) \equiv i(u) + i(v) + \ell \pmod{2}$ by viewing $C$ as a directed path that starts and ends at the same node $u$ and observing that $i(u) + i(u) \pmod{2} = 0$ regardless of the color of $u$.

    \begin{description}
        \item[Base case.] Suppose $\ell = 1$. We claim that $b(P) = a(u, v) \equiv i(u) + i(v) + \ell \pmod{2}$. To see this, we divide the analysis into two cases.
            \begin{itemize}
                \item If $c(u) \neq 3$ and $c(v) \neq 3$, then $a(u, v) \in \{-1, 1\}$, so $a(u, v)  \equiv i(u) + i(v) + \ell  \equiv 1 \pmod{2}$, as $i(u) = i(v) = 0$.
                \item If $c(u) = 3$ or $c(v) = 3$, then $0 = a(u, v) \equiv i(u) + i(v) + \ell  \equiv 0 \pmod{2}$, as $i(u) + i(v) = 1$.
            \end{itemize}
        \item[Summation.] Consider the case where $\ell > 1$.
            \begin{align*}
                b(P) & = \sum_{(s, t) \in P} a(s, t)                                                                                                          \\
                     & \equiv \sum_{(s, t) \in P} (i(s) + i(t) + 1) \pmod{2}                                              & \text{(the base case)}            \\
                     & \equiv - i(u) - i(v) + 2 \cdot \sum_{(s, t) \in P} (i(s) + i(t)) + \sum_{(s, t)  \in P} 1 \pmod{2}                                     \\
                     & \equiv i(u) + i(v) + \ell \pmod{2}.                                                                &                        & \qedhere
            \end{align*}
    \end{description}
\end{proof}

We emphasize that the claim regarding $b(C)$ in \cref{lem:parity} does \emph{not} follow immediately from \cref{lem:cycle}, as \cref{lem:cycle} only considers grids.

Before we continue, we briefly discuss the intuition behind our notion of $b$-value. This notion intuitively captures the level of difficulty in completing the coloring of the rest of the grid given a colored path. In a proper 3-coloring of a grid, the set of nodes with color $3$ separates the remaining nodes into connected regions. In the subsequent discussion, a region refers to a connected set of nodes with colors $1$ and $2$.

Let us start by considering a directed path $3 \rightarrow 2 \rightarrow 1 \rightarrow 2 \rightarrow 1 \rightarrow 2 \rightarrow 3$ whose $b$-value is zero. Starting from the coloring of this path, it is easy to extend the partial coloring in a way that the region of nodes colored with 1 and 2 is enclosed by a cycle of nodes colored with 3, see \Cref{fig:b-value=0}.

\begin{figure}[ht]
    \centering
    \includegraphics[scale=0.55]{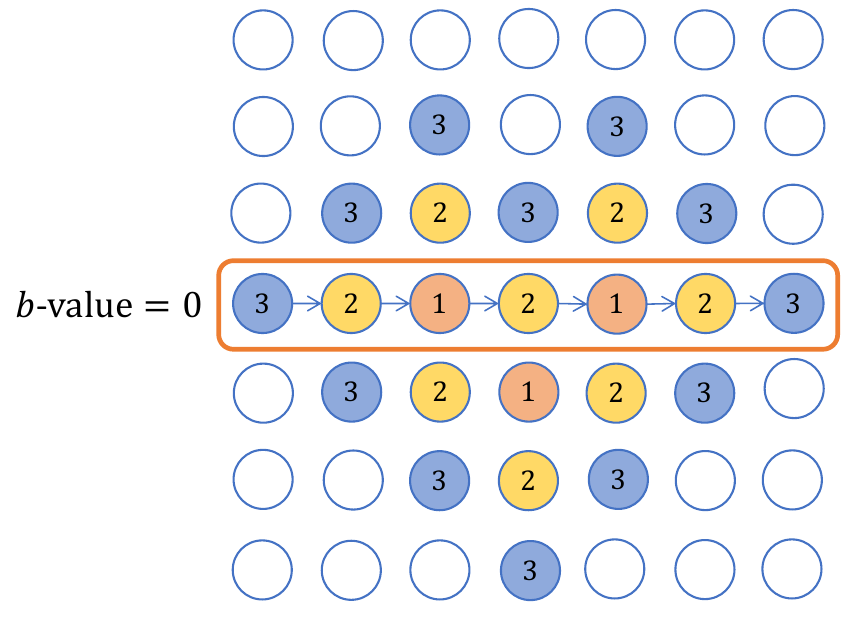}
    \caption{A directed path whose $b$-value is zero.}
    \label{fig:b-value=0}
\end{figure}

Consider a directed path $3 \rightarrow 2 \rightarrow 1 \rightarrow 2 \rightarrow 1 \rightarrow 3$  whose a $b$-value is 1. In this case, we cannot close the region with \emph{one} cycle of nodes colored 3. The region of nodes colored with 1 and 2 must continue unless it closes with itself or reaches the boundary of the grid.

Consider any directed cycle $C$ that contains the directed path  $3 \rightarrow 2 \rightarrow 1 \rightarrow 2 \rightarrow 1 \rightarrow 3$ as a subpath. Observe that $C$ must intersect with the region of nodes colored with 1 and 2 \emph{again} via a directed path whose $b$-value is $-1$, as illustrated in \Cref{fig:b-value=1}.

Intuitively, the $b$-value of a directed path counts the difference between the number of occurrences of $3 \rightarrow 2 \rightarrow \cdots \text{(colors $1$ and $2$ appear alternatively)} \cdots \rightarrow 1 \rightarrow 3$ and the number of occurrences of $3 \rightarrow 1 \rightarrow \cdots \text{(colors $1$ and $2$ appear alternatively)} \cdots \rightarrow 2 \rightarrow 3$. The above informal discussion suggests that each occurrence $3 \rightarrow 2 \rightarrow \cdots \rightarrow 1 \rightarrow 3$ has to be matched with an occurrence of $3 \rightarrow 1 \rightarrow \cdots \rightarrow 2 \rightarrow 3$. Therefore, if we can force an algorithm to create a directed path $P$ with a large value of $|b(P)|$, then we may obtain a high locality lower bound for the algorithm.

\begin{figure}[ht]
    \centering
    \includegraphics[scale=0.55]{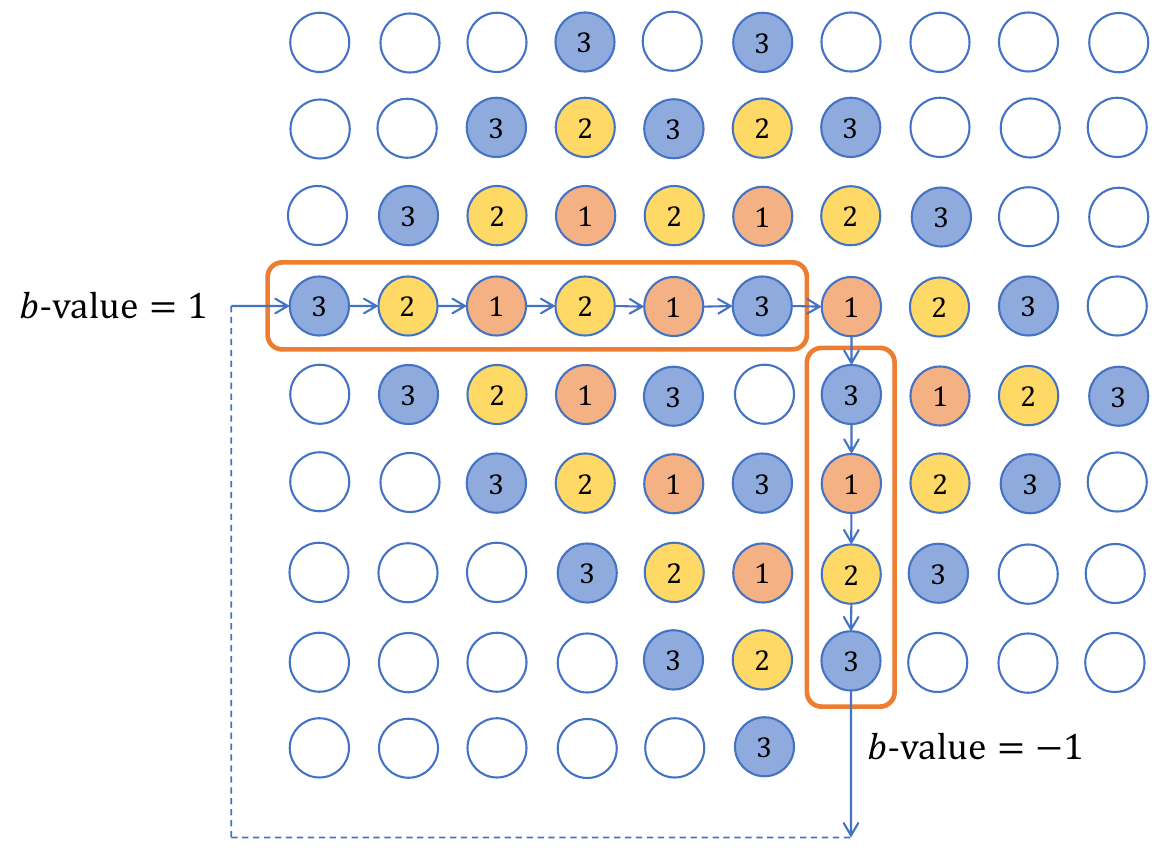}
    \caption{Intersection of a directed cycle and a region of nodes with colors $1$ and $2$.}
    \label{fig:b-value=1}
\end{figure}

\subsection{Hardness of 3-Coloring in Simple Grids}
\label{sec:lower}

In this section, we prove the main result of the paper, i.e., \Cref{thm:complexity}. Our proof uses the following properties of $b$-values in grids.
\begin{itemize}
    \item The $b$-value of any directed cycle is \emph{always zero} (\Cref{lem:cycle}).
    \item The parity of the $b$-value of any directed path is \emph{solely determined} by the parity of the path length and the color of the two endpoints $u$ and $v$ (\Cref{lem:parity}).
\end{itemize}
To establish \Cref{thm:complexity}, we prove that any algorithm designed for 3-coloring grids, operating with a locality of $T(n) \in o(\log n)$, can be strategically countered by an adversary capable of forcing a directed cycle with a non-zero $b$-value. Throughout this section, let $\mathcal{A}$ be any algorithm for 3-coloring a $\left(\sqrt{n} \times \sqrt{n} \right)$ grid $G=(V,E)$ with a locality of $T(n) \in o(\log n)$. We analyze the coloring function $c \colon V \rightarrow \{1, 2, 3\}$ generated by algorithm $\mathcal{A}$.

The core strategy of our proof is to create a directed path with a substantial $b$-value, making it impossible for the algorithm to complete the coloring with a small locality. The interaction between the algorithm and the adversary in $\onlineLOCAL$ can be seen as a 2-player game as follows.

\begin{itemize}
    \item The algorithm's task is to label each node $v_i$ based on the current discovered region $G_i$ and the sequence $(v_1, v_2, \ldots, v_i)$. The algorithm wins if the final coloring of $G$ is proper.
    \item The adversary's task is to select the nodes $v_i$ in the sequence $\sigma=(v_1, v_2, \ldots, v_n)$. Moreover, the adversary has the liberty to adjust how the current discovered region $G_i$ fits into $G$ as an induced subgraph. Informally, suppose $G_i$ consists of several connected components. The adversary has the flexibility to adjust the \emph{directions} of these components and the \emph{distances} between these components, as the algorithm is unaware of the precise location of these components in  $G$. The adversary wins if the final coloring of $G$ is not proper.
\end{itemize}

Next, we show that there is an adversary strategy to construct a directed path with a large $b$-value within a row while keeping the discovered region small. Here we only allow our adversary strategy to select nodes within \emph{one} row, so in the subsequent discussion, we measure the \emph{length} of the current discovered region $G_i$ by the distance between the two farthest nodes in $G_i$ restricted to the row, where the distance is measured with respect to the row.

\begin{lemma}\label{lem:construct}
    Let $k$ be a non-negative integer such that $5^{k+1} \cdot T(n) < \sqrt{n}$.
    There is an adversary strategy to construct a directed path with its $b$-value of at least $k$ within a row while keeping the length of the discovered region at most $5^{k+1} \cdot T(n)$.
\end{lemma}
\begin{proof}
    We prove the lemma by an induction on $k$.

\paragraph{Base case.} For the case where $k = 0$, the adversary can just reveal any node $v$. We view node $v$ as a directed path with a $b$-value of zero. The length of the discovered region is $2 T(n) < 5^{k+1} \cdot T(n)$, as the algorithm has locality $T(n)$.

\paragraph{Inductive step.} Consider the case where $k \geq 1$. In the subsequent discussion, we write $P_{x,y}$ to denote the directed path starting from $x$ and ending at $y$ along the row. By the induction hypothesis, we construct two directed paths $P_{u,v}$ and $P_{s,t}$ with a $b$-value of at least $k-1$. The construction of both directed paths requires a discovered region of length at most $5^{k} \cdot T(n)$.

            We now concatenate these two directed paths together into a directed path $P_{u,t}$ by concatenating the two discovered regions via a path of length $\ell \in \{2,3\}$ on the row. See \Cref{fig:concatenate} for an illustration. After the concatenation, the length of the discovered region is at most $\ell + 2 \cdot 5^{k} \cdot T(n) \leq 5^{k+1} \cdot T(n)$. Here we use the fact that $T(n) \geq 1$, since $3$-coloring cannot be solved with zero locality.

            To finish the proof, we just need to select $\ell \in \{2,3\}$ in such a way that we can find a directed path in the row using the nodes between $u$ and $t$ whose $b$-value is at least $k$. If the $b$-value of one of the paths $P_{u,v}$ and $P_{s,t}$ is already at least $k$, then we are done, so from now on we assume that their $b$-values are precisely $k-1$.

            By \Cref{lem:parity}, the parity of the $b$-value of any directed path is \emph{solely determined} by the parity of the path length and the color of the two endpoints.
            Therefore, we may select $\ell \in \{2,3\}$ in such a way that the parity of the $b$-value of the path $P_{v,s}$  differs from $k-1$, i.e.,~$b(P_{v,s}) \not\equiv k-1 \pmod{2}$.

            From now on, we write $h = b(P_{v,s})$, so we have $b(P_{u,t}) = 2(k-1)+h$. We claim that at least one of $|2(k-1) + h|$ or $|h|$ is greater than $|k-1|$. The claim holds true due to the inequality \[|2(k-1) + h| + |h| = |2(k-1) + h| + |-h| \geq |2(k-1)| = 2 \cdot |k-1|\]
            and the observation that the inequality becomes equality \emph{only} when $|h| = |k-1|$, which is impossible because $h \neq k-1$. We conclude that the $b$-value of at least one of the directed paths $P_{v,s}$, $P_{s,v}$, $P_{u,t}$, and $P_{t,u}$ is at least $k$, as required.
\end{proof}

\begin{figure}[ht]
    \centering
    \includegraphics[scale=0.8]{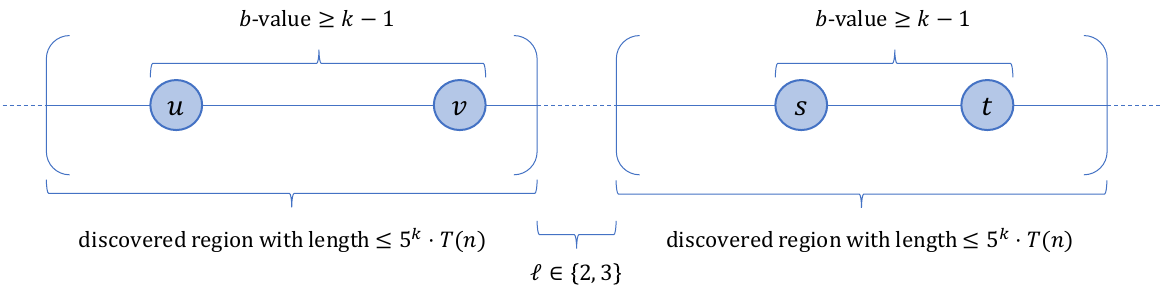}
    \caption{Concatenating $P_{u,v}$ and $P_{s,t}$.}
    \label{fig:concatenate}
\end{figure}

We are prepared to prove \Cref{thm:complexity}.

\main*

\begin{proof}
    To prove the theorem, we just need to show that any algorithm $\mathcal{A}$ with locality $T(n) \in o(\log n)$ for $3$-coloring a $\left(\sqrt{n} \times \sqrt{n}\right)$ grid is incorrect.
    As long as $n$ is at least a sufficiently large constant, we may find an integer $k$ to satisfy the two conditions: $k > 4T(n) + 4$ and $5^{k+1} \cdot T(n) < \sqrt{n}$.

    In the subsequent discussion, if $x$ and $y$ are two nodes that belong to the same row or the same column, then we write $P_{x,y}$ to denote the directed path starting from $x$ and ending at $y$ along the row or the column.

    In the $\left(\sqrt{n} \times \sqrt{n}\right)$ grid, we apply the adversary strategy of \Cref{lem:construct} to force the algorithm $\mathcal{A}$ to construct a directed path $P_{u,v}$ in a row $R_1$ starting from $u$ and ending at $v$ with $b(P_{u,v})\geq k$. After that, we consider the row $R_2$ that is above the row $R_1$ by a distance of $2 T(n) + 2$.
    Let $t$ be the node in $R_2$ that belongs to the same column as $u$.
    Let $s$ be the node in $R_2$ that belongs to the same column as $v$. See \Cref{fig:rectangle} for an illustration.

    As the adversary, we ask the algorithm to color all the nodes in $P_{s,t}$.
    We may assume that $b(P_{s,t}) \geq 0$. In case $b(P_{s,t}) < 0$, we can simply reverse the direction of $P_{s,t}$. This is possible because the discovered region associated with $P_{s,t}$ is not connected to the discovered region associated with $P_{u,v}$ from the viewpoint of the algorithm, as the distance between $R_1$ and $R_2$ is at least $2 T(n) + 2$.

    We let the algorithm finish the coloring of the remaining nodes arbitrarily.
    Now, consider the directed cycle $C$ formed by concatenating the four directed paths $P_{u,v}$, $P_{v,s}$, $P_{s,t}$, and $P_{t,u}$.
    By its definition, the absolute value of the $b$-value of one path $P$ is at most its length, so the $b$-value of both $P_{v,s}$ and $P_{t,u}$ is at least $-(2T(n)+2)$, so
    \[b(C) =b(P_{u,v}) + b(P_{v,s}) + b(P_{s,t}) + b(P_{t,u}) \geq k -(2T(n)+2) + 0 -(2T(n)+2) > 0,\]
    which is impossible due to \Cref{lem:cycle}, so the coloring produced by the algorithm is incorrect.
\end{proof}

We remark that if we consider a general $(a \times b)$ grid, then the above proof yields an $\Omega\left(\min\{ \log \max\{a,b\}, \min\{a,b\}\}\right)$ locality lower bound. To see this, without loss of generality, assume $b \geq a$, and observe that the proof above works so long as the locality $T$ of the algorithm satisfies $T \in o(\log b)$ and $T \in o(a)$.

\begin{figure}[ht]
    \centering
    \includegraphics[scale=0.8]{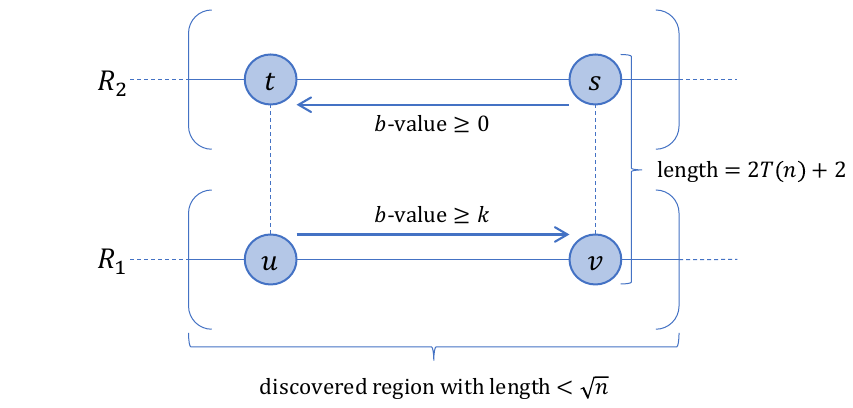}
    \caption{Construction of a directed cycle with non-zero $b$-value}
    \label{fig:rectangle}
\end{figure}

\subsection{Hardness of 3-Coloring in Toroidal and Cylindrical Grids}\label{sec:lower2}

In this section, we prove \Cref{thm:toroidal}.

\torus*

\begin{proof}
    We start by discussing the properties of proper 3-coloring of toroidal and cylindrical grids. Consider any two directed cycles $C_1$ and $C_2$ corresponding to orienting two rows with \emph{different directions}. Same as the proof of \Cref{lem:cycle}, we may calculate $b(C_1) + b(C_2)$ alternatively by summing up the $b$-value for the cells located between $C_1$ and $C_2$, so we infer that
    \begin{equation}
        b(C_1) + b(C_2) = 0. \label{eq:1}
    \end{equation}
    Furthermore, if the number of columns is \emph{odd}, then the lengths of the two directed cycles $C_1$ and $C_2$ are odd, so \Cref{lem:parity} implies that both $b(C_1)$ and $b(C_2)$ are odd numbers.

    Consider an arbitrary algorithm $\mathcal{A}$ designed for 3-coloring $\left(\sqrt{n} \times \sqrt{n} \right)$ cylindrical or toroidal grids with locality $T(n) \in o(\sqrt{n})$. Let us choose a sufficiently large value of $n$ such that  $\sqrt{n}$ is an odd number and $4 T(n) + 4 \leq \sqrt{n}$. Therefore, we may select two rows such that their $T(n)$-radius neighborhoods induce \emph{non-adjacent} cylindrical grids with $2T(n)+1$ rows and $\sqrt{n}$ columns.

    As the adversary, we instruct algorithm $\mathcal{A}$ to color these two rows. From the perspective of the algorithm, it sees precisely two disjoint cylindrical grids with $2 T(n)+1$ rows and $\sqrt{n}$ columns. Furthermore, the algorithm is unaware of their positions and directions in the input graph. As the adversary, we have the
    liberty to choose their directions \emph{after} the algorithm fixes the coloring of the two rows. Since the $b$-value of an odd-length directed cycle is odd, we can always select their directions to violate \Cref{eq:1}. Hence the algorithm $\mathcal{A}$ cannot correctly 3-color the graph.
\end{proof}

We remark that the $\Omega(\sqrt{n})$ locality lower bound in the proof above comes from the fact that the number of rows is $\sqrt{n}$. If we consider a general $(a \times b)$ toroidal or cylindrical grid, then the above proof yields an $\Omega(a)$ locality lower bound whenever the number of columns $b$ is an odd number.

\section{Hardness of \texorpdfstring{$(2k-2)$}{(2k - 2)}-Coloring of \texorpdfstring{$k$}{k}-Partite Graphs}
\label{sec:2k-2}

In this section, we consider the problem of coloring a $k$-partite graph using  $2k-2$ colors and establish a lower bound of $\Omega(n)$ on the locality in the $\onlineLOCAL$ model, where $k \geq 2$ is a constant. Thus we prove \Cref{thm:2k-2}. We introduce a specific type of gadget. Following that, we will proceed with constructing the hard instance graph using these gadgets to establish the lower bound.

\paragraph{Gadget definition.} A gadget $A(k)$ is represented as a graph with a set of nodes denoted by $[k] \times [k]$, where each node corresponds to an ordered pair $(i, j)$. Two nodes, $(i, j)$ and $(i', j')$, are connected by an edge if and only if $i \neq i'$ and $j \neq j'$. Furthermore, for any $i \in [k]$, the set of nodes $\{(i, j) \colon j \in [k]\}$ defines the $i$-th row of the gadget $A(k)$. Similarly, for $j \in [k]$, the set of nodes $\{(i, j) \colon i \in [k]\}$ defines the $j$-th column of the gadget. Therefore, an edge exists between two nodes in a gadget if and only if they are neither in the same row nor the same column.

\paragraph{Construction of the hard instance graph $G^*$.} $G^*$ comprises of $n'=n/k^2$ gadgets denoted as $A_1, \ldots, A_{n'}$, where the node set of the $\ell$-th gadget $A_\ell$ can be represented as $\{\ell \} \times [k] \times [k]$. Consequently, the node set of $G^*$ is $[n'] \times [k] \times [k]$. The edge set of $G^*$ is defined as follows.
\begin{description}
    \item[Edges within the gadgets:] For each $\ell \in [n']$, there is an edge between $(\ell, i, j)$ and  $(\ell, i', j')$ if and only if $i \neq i'$ and $j \neq j'$.
    \item[Edges between consecutive gadgets:] For each $\ell \in [n'-1]$, there is an edge between $(\ell, i, j)$ and $(\ell+1, i', j')$ if and only if $i \neq i'$ and $j \neq j'$.
\end{description}

\begin{proposition}
    \label{prop:k-color}
    The graph $G^*$ is $k$-partite.
\end{proposition}
\begin{proof}
    We will demonstrate a proper coloring of $G^*$ with a set of $k$ colors $\{1,\ldots,k\}$. Recall the construction of $G^*$ that consists of gadgets $A_1,\ldots,A_{n'}$. For any gadget, color all the nodes in the $i$-th row using color $i$, where $i \in [k]$. This coloring is proper due to the following reason: An edge between two nodes, either in the same gadget or in consecutive gadgets, can only exist if they are neither in the same row nor the same column.
\end{proof}

We first sketch the high-level idea behind the lower bound proof.
Given a proper coloring of a gadget, it is said to be \emph{row-colorful} if it has a row where all nodes in the row are colored with distinct colors. Similarly, we can define a \emph{column-colorful} gadget. We show (in \Cref{sec:prop-gadgets}) that, for a proper coloring of a gadget with $2k-2$ colors, the gadget is precisely one out of row-colorful and column-colorful. Furthermore, leveraging the above properties of the gadgets, we establish (in \Cref{sec:prop-hardgraph}) a crucial property of the graph  $G^*$ with gadgets $A_1,\ldots,A_{n'}$: For any proper $(2k-2)$-coloring of $G^*$, all the gadgets are either row-colorful or all of them are column-colorful. To establish the lower bound, consider a scenario where the adversary asks the algorithm to first color the nodes in $A_1$, followed by the nodes in $A_{n'}$, and then the remaining nodes. Regardless of how the algorithm, in the $\onlineLOCAL$ model with locality $o(n)$, colors the nodes in $A_1$ and $A_{n'}$, we argue that the adversary can modify the input graph in a specific way. In particular, the modification results in an input graph that is isomorphic to $G^*$, consistent with the previously explored part of the graph, and crucially ensures that $A_1$ is row-colorful while $A_{n'}$ is column-colorful. However, this leads to a failure to color some nodes in the future, as either all the gadgets in the input graph are row-colorful or all of them are column-colorful.

\paragraph{Comparison with the lower bound proof in the $\LOCAL$ model.} In \cite{coiteux2023no}, the authors showed that $(2k-2)$-coloring a $k$-partite graph requires $\Omega(n)$ rounds in the $\LOCAL$ model. Our hard instance $G^*$, designed to establish the lower bound in the $\onlineLOCAL$ model, is inspired by the hard instance in \cite{coiteux2023no}. Specifically, their hard instance, denoted as $G^{\#}$, comprises $G^*$ and two cliques with $k$ nodes each, denoted as $X$ and $Y$, as its subgraphs. Let the node sets in $X$ and $Y$ be $\{x_1,\ldots,x_k\}$ and $\{y_1,\ldots,y_k\}$, respectively. Furthermore, each $x_i$ is connected to every node in the gadget $A_1$ that is not in the $i$-th row, and each $y_j$ is connected to every node in the gadget $A_{n'}$ that is not in the $j$-th column. It is crucial to highlight that the chromatic number of $G^{\#}$ is $2k-1$. However, the subgraph induced by the $o(n)$ neighborhood of any node is $k$-partite. Note that $G^{\#}$ is a graph outside the input family. However, the instance is locally solvable, meaning the subgraph induced by the $o(n)$-radius neighborhood of any node can be colored with $k$ colors. The lower bound in the $\LOCAL$ model can be established using an indistinguishability argument by \citet{linial1992locality}, which proves lower bounds where the hard instance is taken from outside the input family. While the hard instance $G^*$ in the $\onlineLOCAL$ model draws inspiration from the $G^{\#}$ hard instance developed for the $\LOCAL$ model, the application of Linial's technique, designed for $\LOCAL$, is not straightforward in the $\onlineLOCAL$ model possibly due to the presence of global memory. Despite the inspiration from the $\LOCAL$ model, our arguments to establish the lower bound in the $\onlineLOCAL$ model are different and accommodate the unique features of the $\onlineLOCAL$ setting, including the utilization of global memory. 

\subsection{Properties of the Gadgets}\label{sec:prop-gadgets}
In this section, we discuss some properties of gadgets, specifically establishing that a gadget is exactly one out of row-colorful and column-colorful in a proper coloring with $2k-2$ colors.
We progress towards the proof by introducing the concept of ``confining a color to some row/column'' and discussing a related claim.

\begin{definition}[Confinement of a color]
    Consider a proper coloring of a gadget. A color $c$ is labeled as \emph{confined} to a row (column) if there are at least two nodes within that row (column) colored with the same color $c$.
\end{definition}

\begin{claim}\label{obs:confine}
    Consider a proper coloring of a gadget. The following statements hold.
    \begin{enumerate}
        \item[(i)]  A color can be confined to at most one row (column).
        \item[(ii)]  Moreover, a color cannot be confined to both a row and a column.
    \end{enumerate}
\end{claim}
\begin{proof}
    To establish (i), let us assume that color $c$ is confined to both the $i_1$-th and $i_2$-th rows. In this scenario, there exist four nodes (two in each row) colored with $c$, denoted as $(i_1,j_1)$, $(i_1,j_1')$, $(i_2,j_2)$, and $(i_2,j_2')$. Note that $i_1 \neq i_2$, $j_1 \neq j_1'$, and $j_2 \neq j_2'$. If there is an edge between $(i_1,j_1)$ and $(i_2,j_2)$ (i.e., $j_1 \neq j_2$), then both nodes cannot share the same color $c$. Alternatively, if $j_1 = j_2$ (implying $j_1 \neq j_2'$), there must be an edge between $(i_1,j_1)$ and $(i_2,j_2')$. Hence, coloring them with the same color is not feasible. This leads to a contradiction!

    To establish (ii), let us assume that color $c$ is confined to both the $i$-th row and $j$-th column. Let $(i, j_1)$ and $(i, j_1')$ be nodes in the $i$-th row colored with color $c$. Additionally, there are two nodes in the $j$-th column colored with the same color. Note that, out of these two nodes in the $j$-th column, at least one must not be in the $i$-th row. Let this node be $(i', j)$ with $i' \neq i$. If there is an edge between $(i, j_1)$ and $(i', j)$ (i.e., $j_1 \neq j$), both nodes cannot share the same color $c$. Alternatively, if $j_1 = j$ (implying $j_1' \neq j$), there must be an edge between $(i, j_1')$ and $(i', j)$. Hence, coloring them with the same color is not feasible. This leads to a contradiction!
\end{proof}

Now, let us provide a formal definition of the concept of a gadget being row-colorful or column-colorful. Subsequently, in \Cref{cl:gadget-color}, we will prove the desired property that each gadget is exactly one out of row-colorful and column-colorful.

\begin{definition}[Row-colorful and column-colorful gadgets]
    Consider a proper coloring of a gadget. A row  (column) of the gadget is \emph{colorful} if all its $k$ nodes have distinct colors. In other words, a row (column) is colorful if no color is confined to the row (column). A gadget is categorized as {row-colorful} ({column-colorful}) if it contains at least one colorful row (column).
\end{definition}
\begin{claim}\label{cl:gadget-color}
    \label{clm:no_colorful_column_and_row}
    Let us consider a proper coloring of gadget $A$ with $2k-2$ colors. Then  $A$ is exactly one out of row-colorful and column-colorful.

\end{claim}
\begin{proof}
    We divide the proof into two parts as follows: (i) $A$ is row-colorful or column-colorful, and (ii) $A$ cannot simultaneously be row-colorful and column-colorful.

    To establish (i), assume, by contradiction, that $A$ is neither row-colorful nor column-colorful. This implies that for each row (column), a color is confined to that specific row (column). Applying \Cref{obs:confine}, we find $2k$ colors, each confined to some row or column. However, this scenario implies an impossibility as we have $2k-2$ colors.

    For (ii), assume by contradiction that the $i$-th row of gadget $A$ is colorful, and simultaneously, the $j$-th column of $A$ is also colorful. Consider the set comprising $(2k-1)$ nodes, which are in the $i$-th row or the $j$-th column. Since we are using $2k-2$ colors to color gadget $A$, there must be two nodes from the set of $(2k-1)$ nodes colored with the same color. If these two nodes are in the $i$-th row, it implies that the $i$-th row is not colorful, leading to a contradiction. Similarly, these two nodes cannot be in the $j$-th column. The only possible scenario is that one of them is in the $i$-th row but not in the $j$-th column, and the other is in the $j$-th column but not in the $i$-th row. That is, these two nodes take the form $(i, j_1)$ and $(i_1, j)$, satisfying $i \neq i_1$ and $j \neq j_1$. Due to the construction of gadget $A$, nodes $(i, j_1)$ and $(i_1, j)$ cannot be colored with the same color, as there is an edge between them.
\end{proof}

\subsection{Properties of \texorpdfstring{$G^*$}{G*} and the Lower Bound Proof}\label{sec:prop-hardgraph}
We first show that all gadgets in $G^*$ are of the same category, either row-colorful or column-colorful.
This, in turn, allows us to prove the desired lower bound on the locality of $(2k-2)$-coloring for $k$-partite graphs in \Cref{thm:2k-2}.

\begin{lemma}\label{lem:all-row-all-col}
    Consider the graph $G^*$ with gadgets $A_1,\ldots,A_{n'}$ and a proper coloring of $G^*$ using $2k-2$ colors. Then, either all the gadgets are row-colorful or all of them are column-colorful.
\end{lemma}
\begin{proof}
    By \Cref{clm:no_colorful_column_and_row}, each gadget $A_i$ is exactly one out of row-colorful or column-colorful. Therefore, for any $\ell \in [n'-1]$, it is sufficient to demonstrate that either both $A_\ell$ and $A_{\ell+1}$ are row-colorful or both are column-colorful.

    By contradiction, assume that $A_\ell$ is row-colorful and $A_{\ell+1}$ is column-colorful. Let $i \in [k]$ be such that the $k$ nodes of the $i$-th row of $A_\ell$ are colored with distinct colors. Let $C$ be the set of $k$ colors used by the nodes in the $i$-th row of $A_\ell$. By \Cref{clm:no_colorful_column_and_row}, $A_{\ell+1}$ is not row-colorful. Thus, for each of the $k$ rows of $A_{\ell+1}$, there exists a color confined to the row. Let $c_j$ denote the color confined to the $j$-th row of $A_{\ell+1}$, where $j \in [k]$. In the following paragraph, we show that the colors confined to any of the $k-1$ rows of $A_{\ell+1}$, excluding the $i$-th row of $A_{\ell+1}$, are distinct from the $k$ colors utilized by the nodes in the $i$-th row of $A_\ell$.

    By contradiction, assume that $c_j$ is in $C$ for some $j\neq i$. Let $(i,j_1)$ be the node in the $i$-th row of $A_\ell$ that is colored with color $c_j$. As $c_j$ is confined to the $j$-th row of $A_{\ell+1}$, consider two nodes in it that are colored with color $c_j$. Let the two nodes be $(j,j_2)$ and $(j,j_2')$. Note that $j\neq i$ and $j_2 \neq j_2'$. If there is an edge between $(i, j_1)$ and $(j, j_2)$ (i.e., $j_1 \neq j_2$), both nodes cannot share the same color $c_j$. Alternatively, if $j_1 = j_2$ (implying $j_1 \neq j_2'$), there must be an edge between $(i, j_1)$ and $(j, j_2')$. Hence, coloring them with the same color is not feasible. This leads to a contradiction!

    Recall that we are using $2k-2$ colors. Additionally, as previously argued, for each $j \neq i$, the color $c_j$ is distinct from the $k$ colors present in $C$. Therefore, there are at most $k-2$ colors in the set $\{c_1,\ldots,c_k\} \setminus \{c_i\}$, implying the existence of two rows in $A_{\ell+1}$ where a specific color is confined to both of them. However, this is impossible due to \Cref{obs:confine}.
\end{proof}

We are ready to prove the  $\Omega(n)$ locality lower bound of $(2k-2)$-coloring for $k$-partite graphs.

\higher*

\begin{proof}

    Consider the graph $G^*$ with gadgets $A_1, A_2, \dots, A_{n'}$. As discussed in \Cref{prop:k-color}, $G^*$ is  $k$-partite. For our lower bound proof, we assume that the input graph provided to the algorithm is isomorphic to $G^*$. Notably, the algorithm operates without any knowledge of the gadget information, as well as the row/column details of the nodes.

    Assuming, by contradiction, the existence of an algorithm that properly colors the input graph with $2k-2$ colors in the $\onlineLOCAL$ model with $T \in o(n)=o(n')$ locality. Note that $n \in \Theta(n')$ as $k$ is a constant. Let the adversary ask the algorithm to color the nodes in gadget $A_1$ first, followed by the nodes in gadget $A_{n'}$, and then the remaining nodes. For each $k$, we select a sufficiently large $n$ such that the $T$-radius neighborhood of nodes in gadget $A_1$ is non-adjacent to the $T$-radius neighborhood of nodes in gadget $A_{n'}$. Consequently, when the algorithm colors nodes in $A_1$ and $A_{n'}$, it lacks information about the other gadgets.

    Without loss of generality, assume that when the algorithm colors the nodes in $A_1$, $A_1$ is row-colorful. Now, we discuss how the adversary transforms the input graph into one isomorphic to $G^*$, forcing gadget $A_{n'}$ to be column-colorful.

    Upon coloring the nodes in $A_{n'}$, assume $A_{n'}$ is not column-colorful. By \Cref{clm:no_colorful_column_and_row}, $A_{n'}$ is row-colorful. Consider the set of gadgets in the $T$-radius neighborhood of the nodes in $A_{n'}$. The adversary alters the row/column information of all the nodes in these gadgets, making $A_{n'}$ column-colorful. Importantly, this modification does not affect the $T$-radius neighborhood of nodes in gadget $A_1$. Additionally, the algorithm is oblivious to anything beyond the $T$-radius neighborhood of the nodes in $A_1$ or $A_{n'}$. Consequently, the remaining part of the graph (outside the $T$-radius neighborhood of the nodes in $A_1$ or $A_{n'}$) can be suitably modified to be isomorphic to $G^*$.

    As the algorithm proceeds to color nodes outside the gadgets $A_1$ and $A_{n'}$, it fails to properly color certain nodes. This follows from the fact that either all the gadgets are row-colorful or all of them are column-colorful, as established in \Cref{lem:all-row-all-col}.
\end{proof}

\section{Coloring Graphs with Locally Inferable Unique Colorings}
\label{sec:lpcc}

In this section, we prove \Cref{thm:lpccupper,thm:lpcclower}: For $k \in O(1)$, the locality of $(k + 1)$-coloring $k$-partite graphs with locally inferable unique colorings in the $\onlineLOCAL$ model is $\Theta(\log n)$. The upper bound is shown in \Cref{sec:k+1_algo} and the lower bound is shown in \Cref{sec:k+1_lb}. 

%


\subsection{Upper Bound}
\label{sec:k+1_algo}

In this section, we present an $\onlineLOCAL$ algorithm with $O(\log n)$ locality for $(k+1)$-coloring the graphs in $\mathcal{L}_{k,\ell}$ with $\ell \in O(1)$.

\subsubsection{Previous Approach}
We first review the $O(\log n)$-locality $\onlineLOCAL$ algorithm by \citet{akbari2021locality} that $3$-colors bipartite graphs. 
Consider the color set $\{1,2,3\}$. When a node $u$ is revealed by the adversary, its $T$-radius neighborhood $\mathcal{B}(u, T)$ is \emph{seen} by the algorithm, where $T \in O(\log n)$ is the locality of the $\onlineLOCAL$ algorithm. Following  \citet{akbari2021locality}, each connected component of the subgraph induced by the set of all nodes seen by the algorithm so far is called a \emph{group}.
How the algorithm of \citet{akbari2021locality} colors a newly revealed node $u$ depends on the status of the nodes in $\mathcal{B}(u, T + 1)$ \emph{right before} $u$ is revealed. The algorithm tries to 2-color the nodes using $\{1,2\}$ in most cases and only uses color $3$ in one exceptional case.
\begin{enumerate}
    \item \emph{If all nodes in $\mathcal{B}(u, T + 1)$ are unseen}, then $\mathcal{B}(u, T)$ forms a new group and the algorithm colors $u$ with $1$. Observe that fixing the color of $u$ determines a \emph{unique} way to complete the coloring of $\mathcal{B}(u, T)$ using $\{1,2\}$, so in this sense the \emph{parity} of the group $\mathcal{B}(u, T)$ is fixed.
    \item \emph{If some nodes in $\mathcal{B}(u, T + 1)$ are seen and they are from the same group $C$}, then $\mathcal{B}(u, T)$ is merged into $C$ to form a larger group $C \cup \mathcal{B}(u, T)$. If $u$ is still uncolored, then the algorithm assigns a color from $\{1,2\}$ to $u$ based on the parity of $\mathcal{B}(u, T)$, as this parity determines a {unique} way to complete the coloring of $C \cup \mathcal{B}(u, T)$ using $\{1,2\}$.
    \item \emph{If some nodes in $\mathcal{B}(u, T + 1)$ are seen and they belong to multiple groups $C_1, C_2, \ldots, C_t$}, then $\mathcal{B}(u, T)$ and these groups are merged into a group $C_1 \cup C_2 \cup \cdots \cup C_t \cup \mathcal{B}(u, T)$. The algorithm first unifies the parities of these groups iteratively and then colors $u$ according to the unified parity if $u$ is still uncolored. The unification of the parities of two groups $A$ and $B$ is done as follows.  If their parities are already consistent, then nothing needs to be done. If their parities are inconsistent, then the parity of one group, say $A$, is \emph{flipped} by utilizing color $3$ in three steps.
    \begin{enumerate}
        \item For all uncolored neighbors of all nodes in $A$ colored with $1$, set their colors to be $2$.
        \item For all uncolored neighbors of all nodes in $A$ colored with $2$, set their colors to be $3$.
        \item For all uncolored neighbors of all nodes in $A$ colored with $3$, set their colors to be $1$.
    \end{enumerate}
\end{enumerate}

Intuitively, flipping the parity of a group is achieved by constructing a barrier using the third color, which requires a locality of $3$. For the flipping operation to work correctly, all the nodes colored during the operation must be within the group under consideration. If we always flip the parity of the smaller-sized group when unifying the parities of two groups, then it is guaranteed that the total number of flips performed, from the perspective of a node, is at most $\log n$. Thus, setting the locality of the algorithm to be $T = 3 \log n$ is sufficient.




\subsubsection{Our Algorithm}
Let $k \geq 2$ be any constant. Consider any graph $G = (V,E) \in \mathcal{L}_{k,\ell}$ with $\ell \in O(1)$. To $(k + 1)$-color $G$ with locality $O(\log n)$ in $\onlineLOCAL$, we follow the recipe of the bipartite case described above.

\paragraph{Oracle.} The assumption that  $G \in \mathcal{L}_{k,\ell}$ implies that $G$ admits a \emph{unique} $k$-coloring up to permutation. This can be seen by setting $G' = G$ in \Cref{def:lpcc}.
Moreover, for any connected subset of nodes $C \subseteq V$ of $G$,~\footnote{A subset $C \subseteq V$ is said to be connected subset if the subgraph induced by $C$ is connected.} this unique partition of $V$ into $k$ parts restricted to $C$ can be inferred by the subgraph of $G$ induced by $\mathcal{B}(C, \ell)$, as follows: Take any $k$-coloring of $G[\mathcal{B}(C, \ell)]$ and restrict the coloring to $C$. \Cref{def:lpcc} guarantees that the resulting partition of $C$ is invariant of the choice of the $k$-coloring of $G[\mathcal{B}(C, \ell)]$, so  the partition of $C$ must be consistent with the unique $k$-coloring of $G$ up to a permutation.
Therefore, in the subsequent discussion, we may assume that our $\onlineLOCAL$ algorithm is equipped with an oracle $\mathcal{O}$ such that for any connected subset of nodes $C \subseteq V$ \emph{seen} by the algorithm, the oracle returns a partition $\mathcal{O}(C)$ of $C$ into $k$ parts consistent with the unique $k$-coloring of $G$ up to a permutation. The oracle can be implemented with an extra locality of $\ell \in O(1)$.

\paragraph{Type.} Recall that a group is a connected component of the subgraph induced by the set of nodes seen by the $\onlineLOCAL$ algorithm so far. While oracle $\mathcal{O}$ allows each group $C$ to locally infer the unique $k$-coloring of $G$ restricted to $C$ up to a permutation, different groups can still disagree on the assignment of the $k$ colors $[k]$ to the $k$ parts of the $k$-partite graph $G$. Analogous to the parity for the bipartite case, we call each of the $k!$ possible color assignments a \emph{type}.

\paragraph{Algorithm.} We are prepared to present our $(k+1)$-coloring algorithm given the oracle $\mathcal{O}$. Set $T = 3(k-1) \log n$ 
as the locality of our algorithm. Throughout the algorithm, each group $C$ is associated with a type in such a way that the following induction hypothesis is maintained:
\begin{itemize}
    \item For each colored node $v \in C$ that is adjacent to an uncolored node in $G$, the color of $v$ is consistent with the type of $C$.
\end{itemize}
Consequently, the current partial coloring of $C$ can be completed using the colors in $[k]$ by adapting the color assignment given by the type for the remaining uncolored nodes in $C$.

Similar to the algorithm of \citet{akbari2021locality}, how our algorithm colors a newly revealed node $u$ depends on the status of the nodes in $\mathcal{B}(u, T + 1)$ \emph{right before} $u$ is revealed. The algorithm tries to $k$-color the nodes using $[k]$ in most cases and only uses color $k+1$ in one exceptional case.

We illustrate our algorithm in \Cref{fig:triangular_grid_case_1,fig:triangular_grid_case_2,fig:triangular_grid_case_3,fig:swap} by considering $4$-coloring a triangular grid. The colors in the figures have the following meanings: Gray indicates ``unseen,'' white indicates ``seen,'' and the remaining colors indicate ``seen and committed to a specific color.''

\begin{enumerate}
    \item \emph{If all nodes in $\mathcal{B}(u, T + 1)$ are unseen}, then $\mathcal{B}(u, T)$ forms a new group and the algorithm colors $u$ with $1$. We query the oracle to obtain the partition $\mathcal{O}(\mathcal{B}(u, T))$ and assign a type to the group  $\mathcal{B}(u, T)$ such that the type is consistent with the color of $u$. See \Cref{fig:triangular_grid_case_1} for an illustration.

\begin{figure}[h!]
    \centering
    \includegraphics[width=0.8\linewidth]{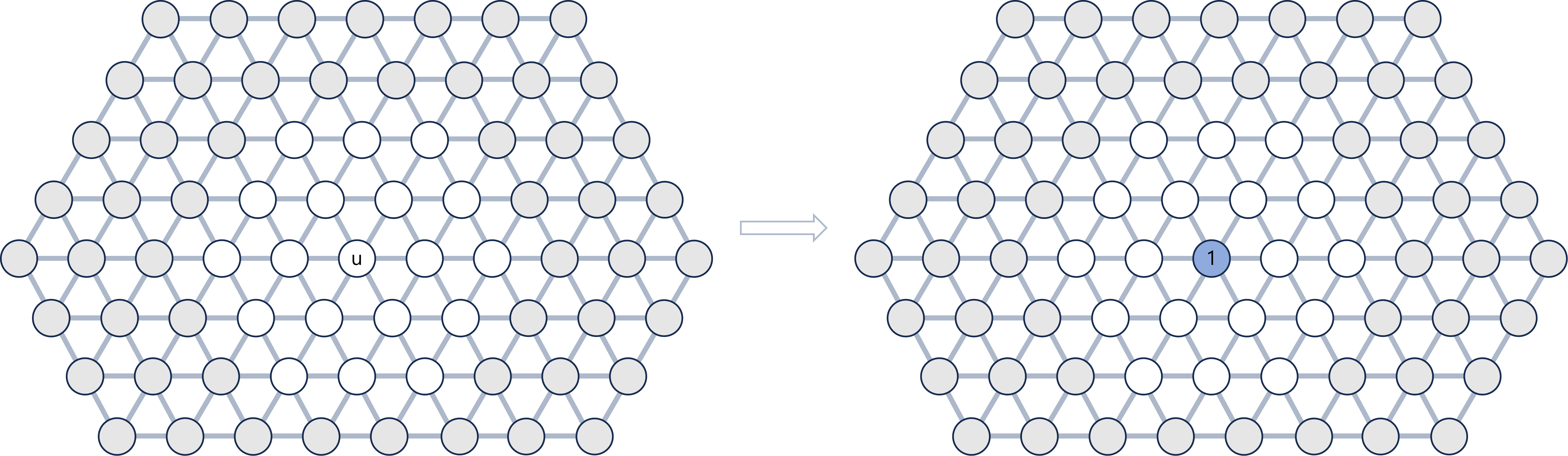}
    \caption{In Case 1, node $u$ in $\mathcal{B}(u, T)$ (the white part) receives color $1$.}
    \label{fig:triangular_grid_case_1}
\end{figure}
    
    \item \emph{If some nodes in $\mathcal{B}(u, T + 1)$ are seen and they are from the same group $C$}, then $\mathcal{B}(u, T)$ is merged into $C$ to form a larger group $C \cup \mathcal{B}(u, T)$. The type of $C \cup \mathcal{B}(u, T)$ is set to be the type of $C$. If $u$ is still uncolored, then we query the oracle to obtain the partition $\mathcal{O}(C \cup \mathcal{B}(u, T))$ and assign a color from $[k]$ to $u$ based on the type of $C$. See \Cref{fig:triangular_grid_case_2} for an illustration.

\begin{figure}[h!]
    \centering
    \includegraphics[width=0.8\linewidth]{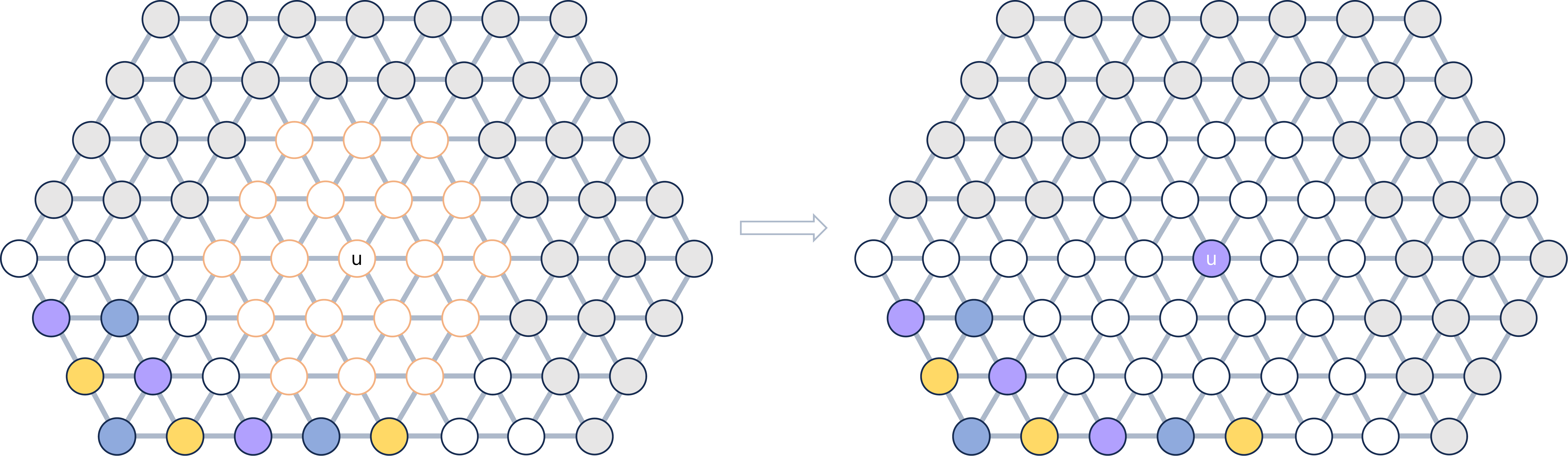}
    \caption{In Case 2, node $u$ in $\mathcal{B}(u, T)$ (the orange part) is colored based on the type of group $C$.}
    \label{fig:triangular_grid_case_2}
\end{figure}

    \item  \emph{If some nodes in $\mathcal{B}(u, T + 1)$ are seen and they belong to multiple groups $C_1, C_2, \ldots, C_t$}, then $\mathcal{B}(u, T)$ and these groups are merged into a group $C_1 \cup C_2 \cup \cdots \cup C_t \cup \mathcal{B}(u, T)$. We query the oracle to obtain the partition $\mathcal{O}(C_1 \cup C_2 \cup \cdots \cup C_t \cup \mathcal{B}(u, T))$, unify the types of these groups, and assign a color from $[k]$ to $u$ based on the unified type if $u$ is still uncolored. The type of the new group $C_1 \cup C_2 \cup \cdots \cup C_t \cup \mathcal{B}(u, T)$ is set to be the unified type.

    The unification of the types of groups $C_1, C_2, \ldots, C_t$ is done as follows. We reorder the groups so that $|C_1| \geq |C_2| \geq \cdots \geq |C_t|$, and then we  \emph{change} the type of each group $X \in \{C_2, \ldots, C_t\}$ to match the type of $C_1$. In the subsequent discussion, let \[\mathcal{O}(C_1 \cup C_2 \cup \cdots \cup C_t \cup \mathcal{B}(u, T)) = V_1 \cup V_2 \cup \cdots \cup V_k.\]
    Therefore, from now on, we may write the type of a group $X \in \{C_2, \ldots, C_t\}$ as a permutation $\pi: [k] \rightarrow [k]$ of the set $[k]$, where $\pi(i) = j$ specifies that the type assigns color $j$ to part $V_i$. 
    
    Given any two distinct colors $i \in [k]$ and $j \in [k]$, \Cref{alg:swap_indices} allows us to swap two colors $i$ and $j$ in the type of $X$. This is done by fixing a certain coloring of the 3-radius neighborhood of the colored subset of $X$ using colors from $[k+1]$. See \Cref{fig:swap} for an illustration of \Cref{alg:swap_indices}.  
    
    Any permutation of $[k]$ can be transformed into any other permutation of $[k]$ by at most $k-1$ swaps, so we can change the type of $X$ to match the type of $C_1$ by executing \Cref{alg:swap_indices} for at most $k-1$ times. See \Cref{fig:triangular_grid_case_3} for an illustration: In this example, the type of one group is changed by swapping yellow and purple by \Cref{alg:swap_indices}, where three layers of colored nodes are created to perform the swap: $\text{yellow} \to \text{orange}$, $\text{purple} \to \text{yellow}$, and $\text{orange} \to \text{purple}$.

\begin{figure}[h!]
    \centering
    \includegraphics[width=0.8\linewidth]{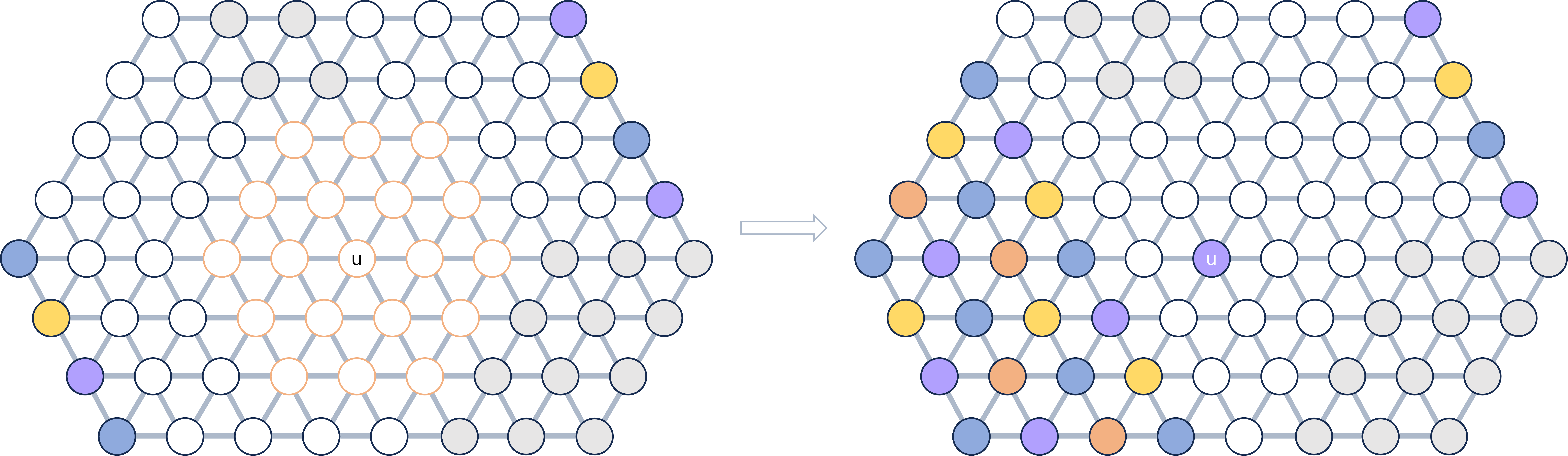}
    \caption{In Case 3,  node $u$ in $\mathcal{B}(u, T)$ (the orange part) is colored based on the type of group $C_1 \cup C_2 \cup \cdots \cup C_t \cup \mathcal{B}(u, T)$ after unifying the types of $C_1, C_2, \ldots, C_t$.}
    \label{fig:triangular_grid_case_3}
\end{figure}
\end{enumerate}


\paragraph{Locality of the algorithm.} Similar to the bipartite case, for our $(k+1)$-coloring algorithm to work correctly, all the nodes colored during each execution of \Cref{alg:swap_indices} must be within the group under consideration so that all the nodes colored by our algorithm are within the seen region. We show that our choice of $T = 3(k-1) \log n$ as the locality of our algorithm is sufficient. 
From the algorithm description, whenever the type of a group $X \in \{C_2, \ldots, C_t\}$ is changed, $X$ is merged into a group $C_1 \cup C_2 \cup \cdots \cup C_t \cup \mathcal{B}(u, T)$ whose size is at least $|C_1| + |X| \geq 2|X|$, as $|C_1| \geq |C_2| \geq \cdots \geq |C_t|$. Therefore, from the perspective of a node, the total number of type changes is at most $\log n$, as the size of a group cannot exceed $n$.
When the type of a group $X$ is changed, \Cref{alg:swap_indices} is executed for at most $k-1$ times, which induces at most $3(k - 1)$ layers of colored nodes. This means that choosing $T = 3(k-1) \log n \in O(\log n)$ ensures that all nodes colored by \Cref{alg:swap_indices} are within the group under consideration.

The overall locality for implementing our algorithm in the $\onlineLOCAL$ model is $T + \ell \in O(\log n)$, where $\ell \in O(1)$ is the cost of implementing the oracle $\mathcal{O}$. For the rest of the section, we prove \Cref{thm:lpccupper} by showing that our algorithm outputs a proper coloring.

\begin{algorithm}[ht]
    \SetKwFunction{change}{change\_index}
    \SetKwProg{func}{Function}{}{}
    \KwIn{A group $X$, a permutation $\pi: [k] \rightarrow [k]$ representing the type of $X$, and two distinct colors $i_1, i_2 \in [k]$.}
    \KwResult{The two colors $i_1$ and $i_2$ are swapped in $\pi$.}
    $X' \gets \text{The nodes of $X$ committed to a color}$\;\nllabel{ln:A'}
    $\change(i_1, k + 1)$\;
    $\change(i_2, i_1)$\;
    $\change(k + 1, i_2)$\;
    \BlankLine
    \func{$\change{i, j}$}{
        \KwIn{A color $i \in [k + 1]$ used in $\pi$ and a color $j \in [k + 1]$ not  used in $\pi$. 
        }
        \KwResult{Color $i$ is replaced by color $j$ in $\pi$.}
        \ForEach{$s \in [k]$ }{
            \eIf{$\pi(s) = i$}{
                \nllabel{ln:change_index}
                Commit each node of $\left( \mathcal{B}(X', 1) \cap V_s\right) \setminus X'$ to color $j$\;
                $\pi(s) \gets j$\;
            }{
                \nllabel{ln:extend_index}
                Commit each node of $\left( \mathcal{B}(X', 1) \cap V_s\right) \setminus X'$ to color $\pi(s)$\;
            }
        }
        $X' \gets \mathcal{B}(X', 1)$\;\nllabel{ln:end}
    }
    \caption{Swapping colors for a group}\label{alg:swap_indices}
\end{algorithm}

\begin{figure}[ht]
    \begin{subfigure}{\linewidth}
        \centering
        \includegraphics[width=0.6\textwidth]{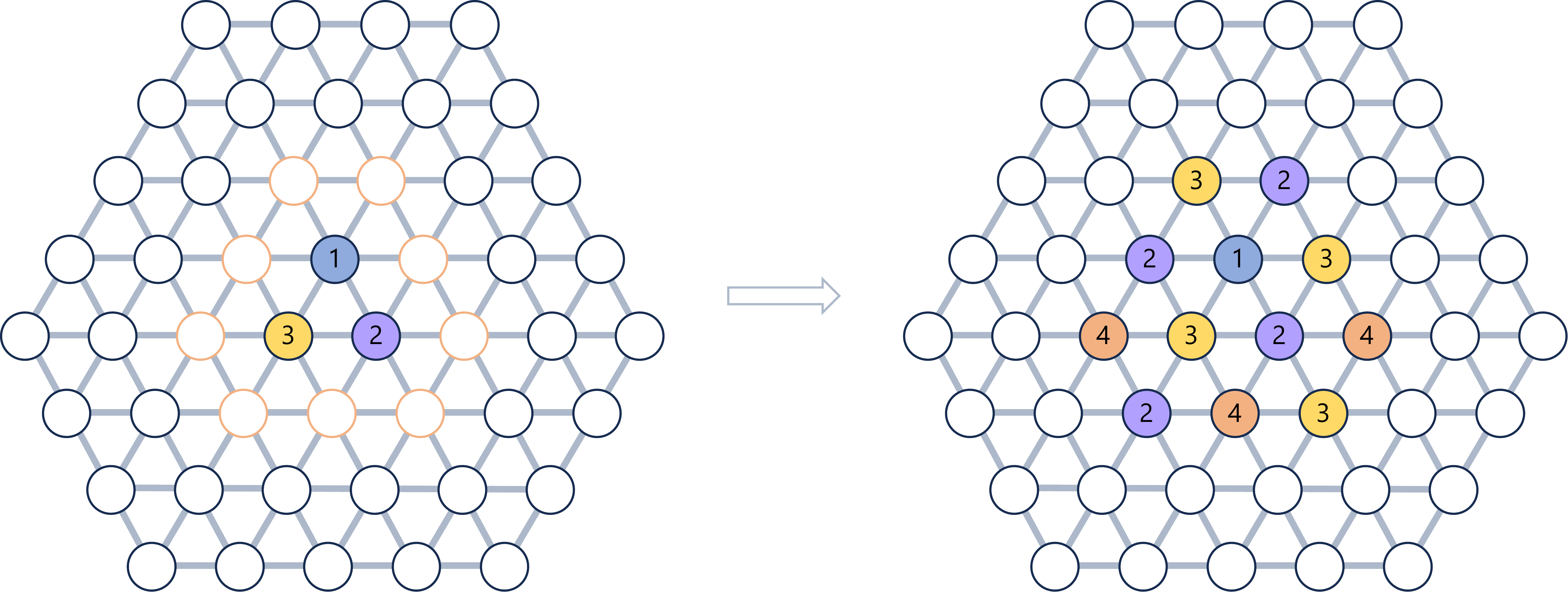}

        \caption{$\change(1, 4)$}
        \label{fig:tri_23}
    \end{subfigure}

    \begin{subfigure}{\linewidth}
        \centering
        \includegraphics[width=0.6\textwidth]{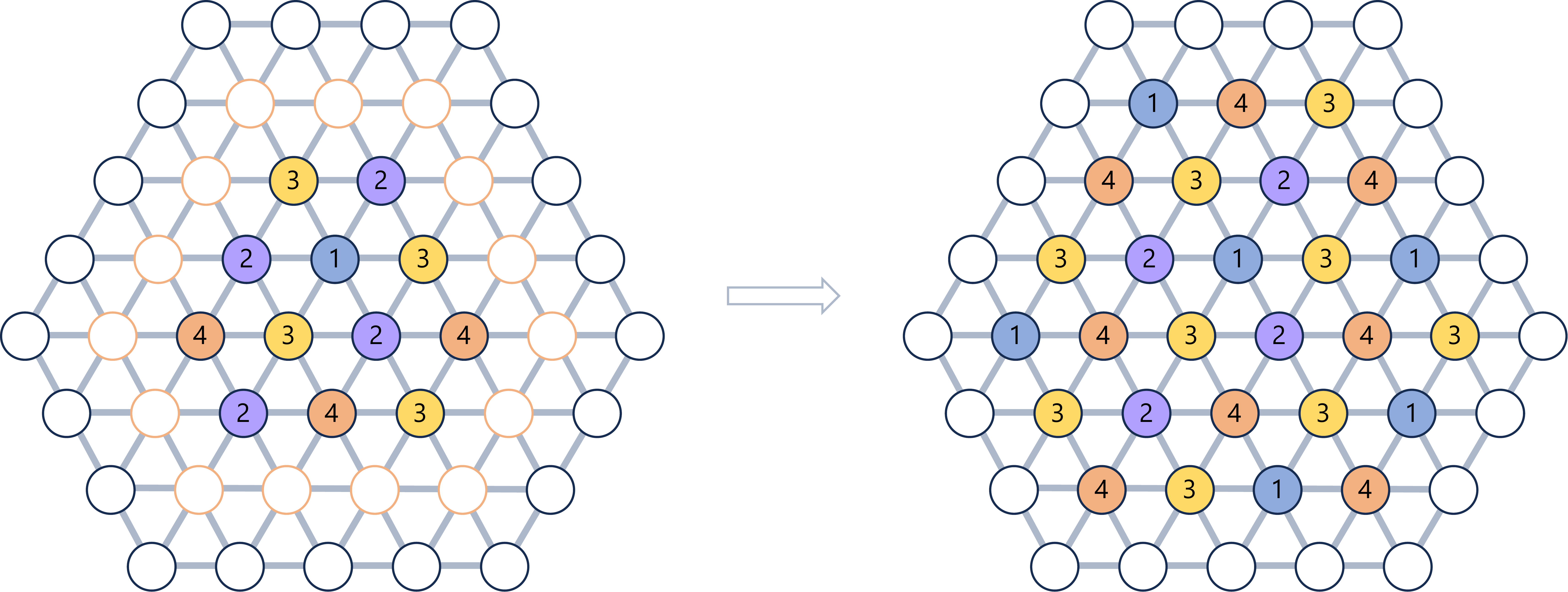}

        \caption{$\change(2, 1)$}
        \label{fig:tri_45}
    \end{subfigure}

    \begin{subfigure}{\linewidth}
        \centering
        \includegraphics[width=0.6\textwidth]{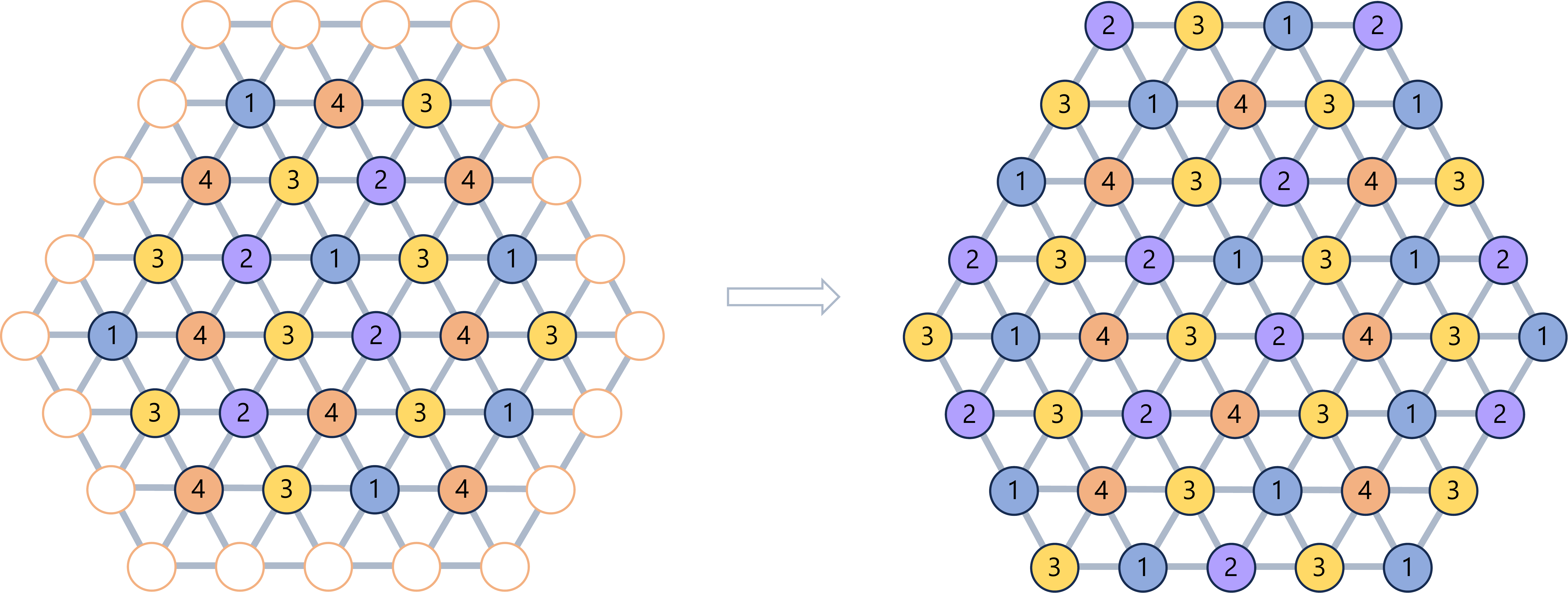}

        \caption{$\change(4, 2)$}
        \label{fig:tri_67}
    \end{subfigure}
    \caption{Swapping colors $1$ and $2$.}\label{fig:swap}
\end{figure}

\lpccupper*
\begin{proof}
Since we already know that our algorithm takes $O(\log n)$ locality in the $\onlineLOCAL$ model, we just need to verify that our algorithm outputs a $(k+1)$-coloring correctly. To do so, we show that the following induction hypothesis is maintained throughout the algorithm for each group $C$.
\begin{itemize}
    \item For each colored node $v \in C$ that is adjacent to an uncolored node in $G$, the color of $v$ is consistent with the type of $C$.
\end{itemize}
Consequently, at any time step, the current partial coloring of $C$ can be completed using the colors in $[k]$ by adapting the color assignment given by the type for the remaining uncolored nodes in $C$. Based on this, we verify that all the color assignments made by the algorithm do not induce a monochromatic edge. We check the three cases in the algorithm description. 
\begin{itemize}
    \item In Case 1, $\mathcal{B}(u, T)$ forms a new group, whose type is chosen to be consistent with the color $1$ assigned to node $u$, so the induction hypothesis is maintained. Since $u$ does not have any colored neighbor, the color assignment cannot induce a monochromatic edge.
    \item In Case 2, $\mathcal{B}(u, T)$ is merged into an existing group $C$ to form a larger group $C \cup \mathcal{B}(u, T)$ that inherits the type of $C$. Since the color assigned $u$ is chosen according to the type of $C \cup \mathcal{B}(u, T)$, the induction hypothesis is maintained and the color assignment does not induce a monochromatic edge.
    \item In Case 3, $\mathcal{B}(u, T)$ and the groups $C_1, C_2, \ldots, C_t$ are merged into a new group $C_1 \cup C_2 \cup \cdots \cup C_t \cup \mathcal{B}(u, T)$. Before the merge, we first change the type of each group $X \in \{C_2, \ldots, C_t\}$ to match the type of $C_1$ by iteratively running \Cref{alg:swap_indices}, so the new group can inherit the type of $C_1$. Similar to Case 2, since the color assigned $u$ is chosen according to the type of $C_1 \cup C_2 \cup \cdots \cup C_t \cup \mathcal{B}(u, T)$, the induction hypothesis is maintained and the color assignment does not induce a monochromatic edge. Thus, we only need to check that \Cref{alg:swap_indices} works correctly. Following the algorithm description, consider \[\mathcal{O}(C_1 \cup C_2 \cup \cdots \cup C_t \cup \mathcal{B}(u, T)) = V_1 \cup V_2 \cup \cdots \cup V_k,\]
    so the type of $X \in \{C_2, \ldots, C_t\}$ can be written as a permutation $\pi$ of the set $[k]$, where $\pi(i) = j$ specifies that the type assigns color $j$ to part $V_i$. Given any two distinct colors $i_1 \in [k]$ and $i_2 \in [k]$, \Cref{alg:swap_indices} swaps the two colors in the color assignment $\pi$ in three steps:  $\change(i_1, k + 1)$, $\change(i_2, i_1)$, and
    $\change(k + 1, i_2)$. To prove the correctness of \Cref{alg:swap_indices}, we just need to verify that $\change(i,j)$ {correctly} replaces color $i$ with color $j$ in $\pi$ in the following sense.
              \begin{enumerate}
                  \item The induction hypothesis is maintained: At the end of $\change(i,j)$, for each colored node $v \in X$ that is adjacent to an uncolored node in $G$, the color of $v$ is consistent with the color assignment $\pi'$ that is the result of replacing $i$ with $j$ in $\pi$. \label{item1}
                  \item \Cref{alg:swap_indices} does not induce monochromatic edges. \label{item2}
              \end{enumerate}
              Let $X'$ be the set of nodes in $X$ that have been colored at the beginning of $\change(i,j)$.
              For each $s \in [k]$, each node $v$ in $(\mathcal{B}(X', 1) \cap V_s)\setminus X'$ is colored with $j$ if $\pi(s) = i$ and $\pi(s)$ otherwise. Since the color assignment is consistent with $\pi'$, \Cref{item1} holds.
              For the case of $\pi(s) \neq i$, by the induction hypothesis, coloring $v$ with $\pi(s)$ does not induce a monochromatic edge. For the case of $\pi(s) = i$, coloring $v$ with $j$ also does not induce a monochromatic edge because color $j$ is not used in $\pi$, so the induction hypothesis guarantees that $v$ does not have a neighbor that is colored with $j$. Therefore, \Cref{item2} holds. \qedhere
\end{itemize}
\end{proof}

\subsection{Lower Bound}\label{sec:k+1_lb}

In this section, we prove that $(k+1)$-coloring of the graphs in $\mathcal{L}_{k,\ell}$ with $\ell \in O(1)$ requires locality $\Omega(\log n)$ in the $\onlineLOCAL$ model, thereby proving \Cref{thm:lpcclower}. 
Recall that $\mathcal{L}_{k,\ell}$ is the class of $k$-partite graphs that admit locally inferable unique colorings with radius $\ell$ (\Cref{def:lpcc}).
For $k=2$, the claim follows from \Cref{thm:complexity}, as $(\sqrt{n} \times \sqrt{n})$ grid belongs to $\mathcal{L}_{2,0}$.


\paragraph{Construction of hard instance graphs $\{G_k\}$.} Given a parameter $n$ such that $\sqrt{n}$ is an integer, we construct a sequence of graphs $G_2, G_3, \ldots$ such that $G_{k} \in \mathcal{L}_{k,\ell}$ and $(k+1)$-coloring of the graphs in $\mathcal{L}_{k,\ell}$ requires locality $\Omega(\log n)$ in the $\onlineLOCAL$ model.
\begin{description}
    \item[Base case:] For $k = 2$, let $G_2$ be the $\left(\sqrt{n} \times \sqrt{n}\right)$ simple grid. Define $H_2 = V(G_2)$.
    \item[Inductive step:] For $k\geq 2$, $G_{k+1}$ is constructed from $G_k$ as follows: For each node $u$ in $G_k$, create a new node $u^\ast$ that is adjacent to $u$ and all the neighbors of $u$ in $G_k$. In the subsequent discussion, we call $u^\ast$ the \emph{duplicate} of $u$ in $G_{k + 1}$. Define $H_{k + 1}$ as the set of the new nodes created in the construction of $G_{k+1}$ from $G_k$. In other words,  $H_{k + 1}$ is the set of all nodes in $G_{k+1}$ that do not belong to $G_{k}$.
\end{description}
Observe that the nodes of $G_{k + 1}$ are partitioned into $k$ layers: $V(G_{k + 1}) = H_2 \cup \dots \cup H_{k + 1}$. We say that a node $v$ is in layer $i$ if $v \in H_i$. We make the following additional observations.

\begin{observation}\label{obs:num_nodes}
    For any  $k\geq 2$, the number of nodes in $G_{k}$ is $n_k=2^{k-2}n$.
\end{observation}
\begin{proof}
The number of nodes in $G_2$ is $n$. From the construction, the number of nodes in $G_{k+1}$ is exactly twice that of $G_{k}$. 
\end{proof}

\begin{observation}
    \label{obs:gk_k_partite}
    $G_k$ is a $k$-partite graph.
\end{observation}
\begin{proof}
A proper $k$-coloring of $G_k$ can be obtained as follows.
We can properly $2$-color $H_2$ using colors $1$ and $2$, as $H_2$ induces a grid. For each layer number $i \in \{3, \ldots, k\}$, we color all nodes 
 in $H_i$ with color $i$. The coloring is proper because $H_i$ is an independent set for each $i \in \{3, \ldots, k\}$.
\end{proof}

\subsubsection{Locally Inferable Unique Coloring}
We first show that $G_k$ admits a locally inferable unique coloring.
Consider any node $v \in H_i$ in layer $i$ of $G_k$, where $i \in \{3, \ldots, k\}$. Since $i \geq 3$, there exists a node $u \in H_2 \cup \cdots \cup H_{i-1}$ such that $v = u^\ast$ is the duplicate of $u$ in the construction of $G_i$. We write $\pi(v) = u$ to denote such a node $u$. By iteratively applying the function $\pi$ to $v$, we eventually obtain a node $w$ in $H_2$. We write $\pi_\diamond(v) = w$ to denote such a node $w$. Intuitively, $\pi(v)$ is the immediate ancestor of $v$, and $\pi_\diamond(v)$ is the ancestor of $v$ in the base layer $H_2$. For each node $v \in H_2$, we define $\pi_\diamond(v) = v$.


    

\begin{claim}\label{clm1}
  For each $v \in V(G_{k})$, there exists a $k$-clique that contains both $v$ and $\pi_{\diamond}(v)$.  
\end{claim}
\begin{proof}
We prove the claim by an induction on $k$. For the base case of $k = 2$, any edge incident to $v = \pi_{\diamond}(v)$ is a desired $k$-clique.
For the inductive step, suppose the claim is true for $k-1$. Consider any $v \in V(G_{k})$. There are two cases.
\begin{itemize}
    \item Suppose $v \in H_{k}$. By the induction hypothesis, $\pi(v)$ and $\pi_{\diamond}(v)$ are contained in a $(k-1)$-clique $K$ of $G_{k-1}$. Since $v$ is adjacent to all neighbors of $\pi(v)$ in $G_k$, $K \cup \{v\}$ is a desired $k$-clique.
    \item Suppose $v \in H_{i}$ for some $i \in \{2, \ldots, k-1\}$. By the induction hypothesis, $v$ and $\pi_{\diamond}(v)$ are contained in a $(k-1)$-clique $K$ of $G_{k-1}$. Observe that the duplicate $v^\ast$ of $v$ in $H_{k}$ is adjacent to every node of $K$. Thus, $K \cup \{v^\ast\}$ is a desired $k$-clique.\qedhere
\end{itemize}
\end{proof}

\begin{claim}\label{clm2}
  For each edge $\{u,v\}$ in $G_{k}$, the two nodes $\pi_{\diamond}(u)$ and $\pi_{\diamond}(v)$ are adjacent.
\end{claim}
\begin{proof}
 Let $u \in H_\ell$ and $v \in H_{\ell^\prime}$. We prove the claim by induction on $\ell + \ell'$. For the base case, when $\ell = \ell' = 2$, $u = \pi_{\diamond}(u)$ is adjacent to $v = \pi_{\diamond}(v)$. Now, consider the inductive step. Since $H_3, H_4, \ldots$ are all independent sets, we must have $\ell \neq \ell'$ unless $\ell = \ell' = 2$. Therefore, without loss of generality, we may assume $\ell < \ell'$. Observe that $\pi(v)$ must be adjacent to $u$ because $v$ is adjacent to $u$ and $v$ is the duplicate of $\pi(v)$. By the induction hypothesis applied to $u$ and  $\pi(v)$, we infer that $\pi_{\diamond}(u)$ and $\pi_{\diamond}(\pi(v)) = \pi_{\diamond}(v)$ are adjacent.
\end{proof}

Let $t \geq 2$. Given two $t$-cliques $K$ and $K'$, we write $K \leftrightarrow K'$ if they share at least $t - 1$ nodes.  
We write $K \overset{\ast}{\leftrightarrow} K'$ if there exists a sequence of $t$-cliques $(K = K_1, \ldots, K_s = K')$ such that $K_1 \leftrightarrow K_2 \leftrightarrow \cdots \leftrightarrow K_{s-1} \leftrightarrow K_s$.  

\begin{claim}\label{clm3}
Let $\{u,v\}$ be any edge in $G_{k}$ such that $u \in H_2$ and $v \in H_2$. For any two $k$-cliques $K$ and $K'$ such that $u \in K$ and $v \in K'$,  $K \overset{\ast}{\leftrightarrow} K'$ in the subgraph of $G_{k}$ induced by  $\mathcal{B}(\{u, v\}, k-1)$.
\end{claim}
\begin{proof}
We prove the claim by induction on $k$.
For the base case of $k = 2$, the claim holds because $K \leftrightarrow \{u,v\} \leftrightarrow K'$ in the 1-radius neighborhood of $\{u,v\}$.
For the inductive step, suppose the claim holds for $k-1$.
Observe that any $k$-clique of $G_k$ contains exactly two nodes from $H_2$ and exactly one node from $H_i$ for each $i \in \{3, \ldots, k\}$, since $H_2$ is triangle-free and $H_i$ is an independent set for each $i \in \{3, \ldots, k\}$.
Therefore, $K_\bullet = K \setminus H_k$ and $K_\bullet' = K' \setminus H_k$ are two $(k-1)$-cliques in $G_{k-1}$.

If $K_\bullet = K_\bullet'$, then we already have $K \leftrightarrow K'$, so from now on we assume $K_\bullet \neq K_\bullet'$.
By the induction hypothesis, there exists a sequence of distinct $(k-1)$-cliques $(K_\bullet = K_1, K_2, \ldots, K_t = K_\bullet')$ in the subgraph of $G_{k-1}$ induced by $\mathcal{B}(\{u, v\}, k-2)$ such that $K_1 \leftrightarrow K_2 \leftrightarrow \cdots \leftrightarrow K_{t-1} \leftrightarrow K_{t}$. 
        
Now, we focus on the two $(k-1)$-cliques $K_i$ and $K_{i + 1}$ for any $i \in [t-1]$. 
\begin{itemize}
    \item Consider the $(k-2)$-clique $\kappa = K_i \cap K_{i + 1}$.
    \item Let $u_1$ be the unique node in  $K_i \setminus \kappa$.
    \item Let $v_1 \in H_{k}$ be the duplicate of $u_1$ in the construction of $G_k$ from $G_{k-1}$ (i.e.,  $\pi(v_1) = u_1$).
    \item Let $u_2$ be the unique node in  $K_{i+1} \setminus \kappa$.
    \item Let $v_2 \in H_{k}$ be the duplicate of $u_2$ in the construction of $G_k$ from $G_{k-1}$ (i.e.,  $\pi(v_2) = u_2$).
    \item Select $w \in H_{k}$ to be the duplicate of any node in $\kappa$. 
\end{itemize}
Observe that $v_1$ is adjacent to $u_1$ and all nodes in $\kappa$, $v_2$ is adjacent to $u_2$ and  all nodes in $\kappa$, and $w$ is adjacent to $u_1$, $u_2$, and all nodes in $\kappa$. Now, consider the following four $k$-cliques (see \Cref{fig:gklk}):
\begin{align*}
   Q_{i,1} = Q_a &= \kappa \cup \{v_1, u_1\},  &&& Q_b &=\kappa \cup \{u_1, w\},\\
   Q_c &=\kappa \cup \{w, u_2\}, &&& Q_{i,2} = Q_d &= \kappa \cup \{u_2, v_2\}.
\end{align*}
We have 
\[Q_{i,1} =  Q_a \leftrightarrow Q_b \leftrightarrow Q_c \leftrightarrow Q_d = Q_{i,2},\]
and these $k$-cliques have the following properties.
        \begin{itemize}
            \item $Q_{i, 1}$ is the union of $K_i$ and the duplicate of a node in $K_i$.
            \item $Q_{i, 2}$ is the union of $K_{i + 1}$ and the duplicate of a node in $K_{i + 1}$.
            \item All the nodes in $Q_a$, $Q_b$, $Q_c$, and $Q_d$ are confined to $\mathcal{B}(\{u, v\}, k-1)$.
        \end{itemize}
        \begin{figure}
            \centering
            \includegraphics[width=0.35\linewidth]{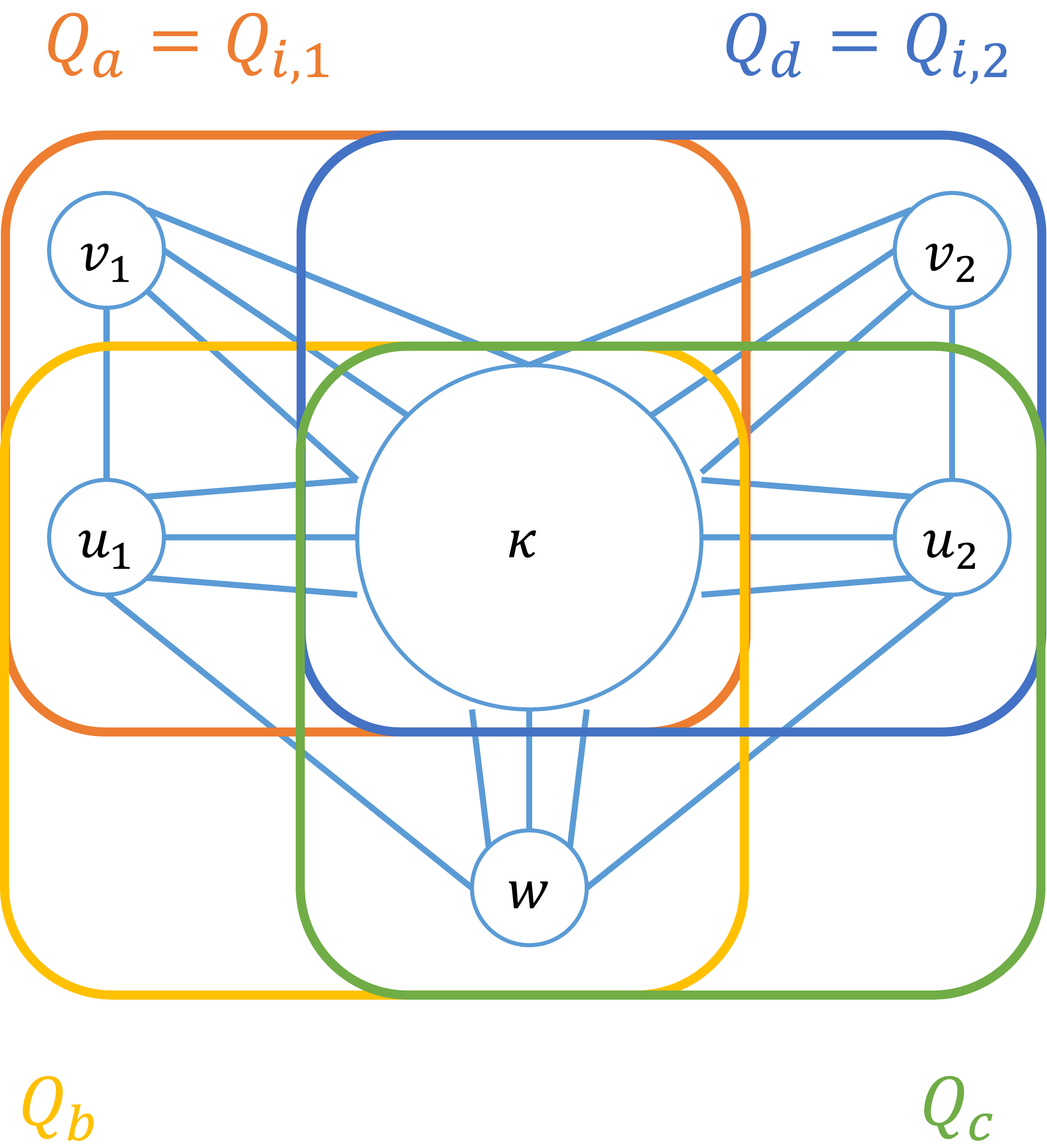}
            \caption{The $k$-cliques $Q_a$, $Q_b$, $Q_c$, and $Q_d$.}
            \label{fig:gklk}
        \end{figure}
By concatenating $K$, the sequence of $k$-cliques $Q_{i,1} =  Q_a \leftrightarrow Q_b \leftrightarrow Q_c \leftrightarrow Q_d = Q_{i,2}$ from $i=1$ to $i=t-1$, and $K'$ into a sequence, we obtain that
        \[
            K \leftrightarrow Q_{1, 1} \overset{\ast}{\leftrightarrow} Q_{1, 2} \leftrightarrow Q_{2, 1} \overset{\ast}{\leftrightarrow} Q_{2, 2} \leftrightarrow \cdots \leftrightarrow Q_{t- 1, 1} \overset{\ast}{\leftrightarrow} Q_{t- 1, 2} \leftrightarrow K', 
        \]
        and all of the $k$-cliques in the sequence are confined to $\mathcal{B}(\{u, v\}, k-1)$.
\end{proof}

We are ready to show that $G_k$ admits a locally inferable unique coloring.

\begin{lemma}\label{lem:perm-cons}
    For any constant $k\geq 2$,  $G_k \in \mathcal{L}_{k,\ell}$ with $\ell \in O(1)$.
\end{lemma}
\begin{proof}
Set $\ell = k \in O(1)$.
  Given a connected subgraph $G'=(V', E')$ of $G_k$, we show that all $k$-colorings of the $\ell$-radius neighborhood of $V'$ in $G_k$ restricting to $V'$ yield the \emph{same} $k$-coloring of $V'$ up to permutation. Consider any proper $k$-coloring of the $\ell$-radius neighborhood of $V'$ in $G_k$, and let $u$ and $v$ be two arbitrary nodes in $G'$. We will find a $k$-clique $K$ with $u \in K$ such that the color of $v$ is \emph{uniquely determined} given the coloring of $K$.

  Since $G'$ is connected, there exists a $u$-$v$ path $(u = w_1, w_2, \ldots, w_t = v)$ in $G'$. By \Cref{clm1}, there exists a  $k$-clique $K_i$ in $G_k$ such that $\{w_i, \pi_{\diamond}(w_i) \} \subseteq K_i$ for each $i \in [t]$.  By \Cref{clm2},  $\pi_{\diamond}(w_i)$ and $\pi_{\diamond}(w_{i + 1})$ are adjacent in $G_k$ for each $i \in [t-1]$. By \Cref{clm3}, we have
  \[K_1 \overset{\ast}{\leftrightarrow}  K_2 \overset{\ast}{\leftrightarrow} 
 \cdots  \overset{\ast}{\leftrightarrow} K_{t-1} \overset{\ast}{\leftrightarrow} K_t\]
  in the subgraph of $G_k$  induced by the $(k-1)$-radius neighborhood of $\{\pi_{\diamond}(w_1), \pi_{\diamond}(w_2), \ldots, \pi_{\diamond}(w_t)\}$, which is within the $k$-radius neighborhood of the $u$-$v$ path $(u = w_1, w_2, \ldots, w_t = v)$.

  For any two $k$-cliques $Q$ and $Q'$ such that $Q \leftrightarrow Q'$, fixing the $k$-coloring of $Q$ uniquely determines the $k$-coloring of $Q'$. Thus, fixing the coloring of the first clique $K_1$, which contains $u$, uniquely determines the coloring of the last clique $K_{t}$, which contains $v$.
\end{proof}

\subsubsection{Locality Lower Bound}
Next, we establish a lower bound on the locality of $(k+1)$-coloring $G_k$.
\begin{lemma}\label{lem:k+1-k-lower}
    Let $k \geq 2$. In the $\onlineLOCAL$ model, the locality of $(k+1)$-coloring $G_k$ is $\Omega\left( \log \frac{n_k}{2^{k-2}} \right)$, where $n_k$ is the number of nodes in $G_{k}$.
\end{lemma}

\begin{proof}
    We prove the lemma by an induction on $k$.
    For the base case of $k = 2$, $G_2$ is the $\left(\sqrt{n} \times \sqrt{n}\right)$ simple grid, so the lower bound $\Omega\left( \log \frac{n_k}{2^{k-2}} \right) = \Omega(\log n)$ immediately follows from \Cref{thm:complexity}.

    For the inductive step, suppose the lower bound $\Omega\left( \log \frac{n_k}{2^{k-2}} \right)$ has been established for $G_{k}$. To prove the lemma by contradiction, suppose there exists an algorithm $\mathcal{A}$ that can properly color $G_{k+1}$ using colors in $[k+2]$ and with locality $T \in o\left( \log \frac{n_{k+1}}{2^{k-1}} \right)$.


    Using $\mathcal{A}$ as a black box, we demonstrate an algorithm  $\mathcal{A}'$ to properly color the nodes in $G_k$ with colors in $[k+1]$, as follows. When a node $u$ in $G_k$ is revealed by the adversary, we first ask the algorithm $\mathcal{A}$ to color the same node $u$ in $G_{k+1}$. Let $c \in [k + 2]$ be the color chosen for $u$ in $G_{k+1}$. If $c \in [k+1]$, then we color node $u$ in $G_k$ with $c$. If $c = k + 2$, then we ask the algorithm $\mathcal{A}$ to color its duplicate $u^\ast$ in $G_{k+1}$ and color node $u$ in $G_k$ with the same color $c'$ used by $u^\ast$. Observe that $c' \in [k+1]$, as $u$ in $G_{k+1}$ has already been colored $c = k+2$ by  $\mathcal{A}$. Therefore, $\mathcal{A}'$ only uses colors from $[k+1]$.

    
    We show that $\mathcal{A}'$ is correct. 
    Suppose $\mathcal{A}'$ produces a monochromatic edge $\{u,v\}$ in $G_k$.
    Let $c$ be the color used by both $u$ and $v$.
    By the algorithm description, we know that the colors of $u$ and $v$ in $G_{k+1}$ given by $\mathcal{A}$  must come from $\{c, k+2\}$. By the correctness of $\mathcal{A}$,  $\{u,v\}$ cannot be monochromatic in the coloring of $G_{k+1}$ by $\mathcal{A}$. Therefore, without loss of generality, we may assume that $u$ is colored $c$ and  $v$ is colored $k+2$ in $G_{k+1}$ by $\mathcal{A}$.  By the algorithm description, the duplicate $v^\ast$ of $v$ must be colored $c$ in $G_{k+1}$ by $\mathcal{A}$, so $\{u, v^\ast\}$ is a monochromatic edge, contradicting the correctness of  $\mathcal{A}$. Therefore, $\mathcal{A}'$ outputs a proper $(k+1)$-coloring of $G_k$.

     
     The locality of algorithm $\mathcal{A}'$ is the same as that of algorithm $\mathcal{A}$, which is $T \in o\left( \log \frac{n_{k+1}}{2^{k-1}} \right)=o\left( \log \frac{n_{k}}{2^{k-2}} \right)$, due to \Cref{obs:num_nodes}, contradicting the induction hypothesis, so we have the desired locality lower bound $\Omega\left( \log \frac{n_{k+1}}{2^{k-1}} \right)$ for $G_{k+1}$. 
\end{proof}

We are ready to prove \Cref{thm:lpcclower}.

\lpcclower*
\begin{proof}
 By \Cref{lem:perm-cons}, $G_k \in \mathcal{L}_{k,\ell}$ with $\ell \in O(1)$. Hence the theorem follows from \Cref{lem:k+1-k-lower} along with the assumption that $k$ is a constant.
\end{proof}


\printbibliography

\end{document}